\documentclass{article}
\usepackage{caption}
\usepackage{amsmath,amsthm}
\usepackage{hyperref}
\usepackage[capitalise]{cleveref}
\usepackage{etoolbox}
\usepackage{textpos}
\newtheorem{theorem}{Theorem}[section]

\newtheorem{lemma}[theorem]{Lemma}
\usepackage
[
        a4paper,
        left=1in,
        right=1in,
        top=1in,
        bottom=1in
]
{geometry}
\usepackage{fancyhdr}
\fancypagestyle{fancyfrong}{
\fancyhf{}
\rhead{}
\lhead{}
\lfoot{
}
}
\usepackage{microtype}
\usepackage{graphicx}
\usepackage{caption}
\usepackage{subcaption}
\usepackage{float}
\usepackage{multirow}
\usepackage[ruled,linesnumbered]{algorithm2e}
\usepackage{mathtools}
\usepackage{tikz}
\usetikzlibrary{quantikz2}
\usetikzlibrary{calc}
\usepackage{framed}
\usepackage{tabularx}
\newlength{\RoundedBoxWidth}
\newsavebox{\GrayRoundedBox}
\newenvironment{GrayBox}[1]%
{\setlength{\RoundedBoxWidth}{.93\columnwidth}
	\def\boxheading{#1}
	\begin{lrbox}{\GrayRoundedBox}
		\begin{minipage}{\RoundedBoxWidth}}%
		{   \end{minipage}
	\end{lrbox}
	\begin{center}
		\begin{tikzpicture}%
			\node(Text)[draw=black!20,fill=white,rounded corners,inner xsep=2ex,inner ysep=2ex,text width=\RoundedBoxWidth]
			{\usebox{\GrayRoundedBox}};
			\coordinate(x) at (current bounding box.north west);
			\node [draw=white,rectangle,inner sep=3pt,anchor=north west,fill=white]
			at ($(x)+(6pt,.75em)$) {\boxheading};
		\end{tikzpicture}
\end{center}}
\newenvironment{defproblemx}[2]{\noindent\ignorespaces%
	\FrameSep=6pt%
	\parindent=6pt%
	\vspace{-3mm}            
	\begin{GrayBox}{#1}%
		\begin{tabular*}{\columnwidth}{!{\extracolsep{\fill}}@{\hspace{.1em}} >{\itshape} p{#2} p{0.84\columnwidth} @{}}%
		}{\\[-1.5ex]
		\end{tabular*}%
	\end{GrayBox}%
	\ignorespacesafterend
	\vspace{-4mm}
}
\newcommand{\problemQuestion}[3]{%
	\begin{defproblemx}{#1}{1.5cm}
		Input: & #2 \\
		Output: & #3
	\end{defproblemx}
}

\DeclareMathOperator*{\argmin}{arg\,min}
\usepackage{amsfonts,amsmath,amssymb,amsthm}
\usepackage{soul}
\usepackage{xcolor}
\usepackage{xspace}
\usepackage{paralist}
\usepackage{complexity}
\usepackage{paralist}

\graphicspath{{figs/}{./figs/}}
\Crefname{observation}{Observation}{Observations}
\Crefname{algorithm}{Algorithm}{Algorithms}
\Crefname{section}{Sect.}{Sects.}
\Crefname{observation}{Observation}{Observations}
\Crefname{lemma}{Lemma}{Lemmas}
\Crefname{corollary}{Corollary}{Corollaries}
\Crefname{claimx}{Claim}{Claims}
\Crefname{figure}{Fig.}{Figs.}
\Crefname{figure}{Fig.}{Figs.}
\Crefname{invariant}{Inv.}{Invs.}
\Crefname{enumi}{Condition}{Conditions}
\Crefname{property}{Property}{Properties}
\Crefname{assumption}{Assumption}{Assumptions}

\newtheorem{claimx}{Claim}

\definecolor{realblue}{rgb}{0,0,1}
\definecolor{lipicsblue}{rgb}{0.08235294118,0.3098039216,0.537254902}
\usepackage{marginnote}
\definecolor{realred}{rgb}{1,0,0}
\definecolor{darkerblue}{rgb}{0.094,0.455,0.804}
\definecolor{darkblue}{rgb}{0.063,0.306,0.545}
\definecolor{red}{rgb}{0.627,0.117,0.156}
\definecolor{green}{rgb}{0,0.588,0.509}
\definecolor{orange}{rgb}{0.903,0.739,0.382}

\SetKwInput{KwInput}{Input}
\SetKwInput{KwOutput}{Output}
\SetNlSty{}{}{}

\SetCommentSty{mycommfont}
\let\oldnl\nl%
\newcommand\nonl{%
\renewcommand{\nl}{\let\nl\oldnl}}
\hypersetup{colorlinks=true,urlcolor=blue,citecolor=blue,linkcolor=blue}

\newcommand{\qedclaim}{\hfill $\blacksquare$}

\providecommand{\keywords}[1]{\textbf{\textit{Keywords: }} #1}
\usepackage{authblk}

\begin{document}
\title{Quantum Graph Drawing}

\author{Susanna Caroppo}
\author{Giordano {Da Lozzo}}
\author{Giuseppe Di Battista}
\affil{\em Roma Tre University \\ \href{mailto:susanna.caroppo@uniroma3.it,giordano.dalozzo@uniroma3.it,giuseppe.dibattista@uniroma3.it}{\{susanna.caroppo,giordano.dalozzo,giuseppe.dibattista\}@uniroma3.it}}
\date{}

\maketitle

\begin{abstract}
In this paper, we initiate the study of quantum algorithms in the Graph Drawing research area. We focus on two foundational drawing standards: {\em 2-level drawings} and {\em book layouts}. Concerning $2$-level drawings, we consider the problems of obtaining drawings with the minimum number of crossings, $k$-planar drawings, quasi-planar drawings, and the problem of removing the minimum number of edges to obtain a $2$-level planar graph. Concerning book layouts, we consider the problems of obtaining $1$-page book layouts with the minimum number of crossings, book embeddings with the minimum number of pages, and the problem of removing the minimum number of edges to obtain an outerplanar graph.
We explore both the quantum circuit and the quantum annealing models of computation. In the {\em quantum circuit model}, we provide an algorithmic framework based on Grover's quantum search, which allows us to obtain, at least, a quadratic speedup on the best classical exact algorithms for all the considered problems. In the {\em quantum annealing model}, we perform experiments on the quantum processing unit provided by D-Wave, focusing on the classical $2$-level crossing minimization problem, demonstrating that quantum annealing is competitive with respect to classical algorithms.

\smallskip
\keywords{
Quantum complexity, 
Grover's algorithm,
QUBO,
D-Wave,
2-Level drawings, 
Book layouts
}
\end{abstract}

\clearpage
\tableofcontents
\clearpage

\section{Introduction}
\pagestyle{fancyfrong}

In this paper, we initiate the study of quantum algorithms in the Graph Drawing research~area. We focus on two foundational graph drawing standards: {\em 2-level drawings} and {\em book layouts}. 
In a \emph{2-level drawing}, the graph is bipartite, the vertices are placed on two horizontal lines, and the edges are drawn as $y$-monotone curves.
In this drawing standard, we consider the search version of the {\sc Two-Level Crossing Minimization} (TLCM) problem, where given an integer $\rho$ we seek a $2$-level drawing with at most $\rho$ crossings, and of the {\sc Two-Level Skewness} (TLS) problem, where given an integer $\sigma$ we seek to determine a set of $\sigma$ edges whose removal yields a \emph{2-level planar graph}, i.e., a forest of caterpillars~\cite{DBLP:conf/acsc/EadesMW86}. 
The minimum value of $\sigma$ is the \emph{2-level skewness} of the considered graph.
We also consider the {\sc Two-Level Quasi Planarity} (TLQP) problem, where we seek a drawing in which no three edges pairwise cross, i.e., a \emph{quasi-planar} drawing, and the {\sc Two-Level $k$-Planarity} (TLKP) problem, where we seek a drawing in which each edge participates to at most $k$ crossings, i.e., a \emph{$k$-planar} drawing.
In a \emph{book layout}, the drawing is constructed using a collection of half-planes, called \emph{pages}, all having the same line, called \emph{spine}, as their boundary. The vertices lie on the spine and each edge is drawn on a page. In this drawing standard, we consider the search version of the {\sc One-Page Crossing Minimization} (OPCM) problem, where given an integer $\rho$ we seek a $1$-page layout with at most $\rho$ crossings; the {\sc Book Thickness} (BT) problem, where we search a $\tau$-page layout where the edges in the same page do not cross, i.e., a \emph{$\tau$-page book embedding}; and the {\sc Book Skewness} (BS) problem, where given an integer $\sigma$ we seek a set of $\sigma$ edges whose removal yields a graph admitting a \emph{$1$-page book embedding}, i.e., it is outerplanar~\cite{DBLP:journals/jct/BernhartK79}.
The minimum value of $\sigma$ is the \emph{book skewness} of the considered~graph.

\paragraph{Our contributions.} We delve into both the \emph{quantum circuit}~\cite{DBLP:books/daglib/0046438,10.5555/1973124} and the \emph{quantum annealing}~\cite{DBLP:series/synthesis/2014McGeoch} models of computation.
In the former, quantum gates are used to compose a circuit that transforms an input superposition of qubits into an output superposition. The circuit design depends on both the problem and the specific instance being processed. The output superposition is eventually measured, obtaining the solution with a certain probability. The quality of the circuit is measured in terms of its \emph{circuit complexity}, which is the number of elementary gates it contains, of its \emph{depth}, which is the maximum number of a chain of elementary gates from the input to the output, and of its \emph{width}, which is the maximum number of elementary gates  ``along a cut'' separating the input from the output. 
It is natural to upper bound the time complexity of the execution of a quantum circuit either by its depth, assuming the gates at each layer can be executed in parallel, or by its circuit complexity, assuming the gates are executed sequentially. The width estimates the desired level of parallelism. 
In the latter, quantum annealing processors, in general quite different from those designed for the quantum circuit model, consist of a fixed-topology network, whose vertices correspond to qubits and whose edges correspond to possible interactions between qubits. A problem is mapped to an embedding on such a topology. During the computation, the solution space of a problem is explored, searching for minimum-energy states, which correspond to, in general approximate, solutions. 

In the quantum circuit model, we first show that the above graph drawing problems can be described by means of quantum circuits. To do that, we introduce several efficient elementary circuits, that can be of general usage in {\em Quantum Graph Drawing.} 
Second, we present an algorithmic framework based on Grover's quantum search approach.
This framework enables us to achieve, at least, a quadratic speedup compared to the best exact classical algorithms for all the problems under consideration. \cref{tab:circuit-model-results} overviews our complexity results and compares them with exact algorithms.

In the quantum annealing model, we focus on the processing unit provided by D-Wave, which allows us to perform \emph{hybrid computations}, i.e., computations that are partly classical and partly quantum. We first show that it is relatively easy to use D-Wave for implementing heuristics for the above problems. Second, we focus on the classical TLCM problem. Through experiments, we demonstrate that quantum annealing exhibits competitiveness when compared to classical algorithms. \cref{tab:d-wave-experiment} overviews our experimental findings.

\begin{table}[tb!]
    \renewcommand{\arraystretch}{1.7}
    \centering
    \caption{Results presented in this paper and comparison with exact algorithms. FPT algorithms are mentioned, if any, only with respect to the natural parameter. CC stands for Circuit Complexity. $M$ denotes the number of solutions to the considered problem.}
\resizebox{\columnwidth}{!}{%
    \begin{tabular}{|c|c|c|c|c|c|c|c|}
        \hline
        \multicolumn{1}{|c|}{\multirow{2}{*}{\begin{minipage}[t][0.9cm][c]{1.2cm} \centering Problem \end{minipage} }} & 
        \multicolumn{1}{c|}{\multirow{2}{*}{\begin{minipage}[t][0.9cm][c]{2.5cm} \centering Classic Algorithm Running Time \end{minipage} }} & 
        \multicolumn{1}{c|}{\multirow{2}{*}{\begin{minipage}[t][0.9cm][c]{2cm} \centering Upperbound for $m$  \end{minipage} }} & 
        \multicolumn{1}{c|}{\multirow{2}{*}{\begin{minipage}[t][0.9cm][c]{1.5cm} \centering FPT Time\end{minipage} }} & 
        \multicolumn{1}{c|}{\multirow{2}{*}{\begin{minipage}[t][0.9cm][c]{2.4cm} \centering \mbox{Quantum Oracle} Calls \end{minipage}}} & 
        \multicolumn{3}{c|}{Oracles~(\cref{lem:oracles})}                                           \\ \cline{6-8}
        \multicolumn{1}{|c|}{}                   & 
        \multicolumn{1}{c|}{}                   & 
        \multicolumn{1}{c|}{}                   & 
        \multicolumn{1}{c|}{}                   & 
        \multicolumn{1}{c|}{}                   & 
        \multicolumn{1}{c|}{\begin{minipage}[c][0.9cm][c]{1.3cm} \centering CC \end{minipage}  } & 
        \multicolumn{1}{c|}{\begin{minipage}[c][0.9cm][c]{1.6cm} \centering Depth \end{minipage}} & 
        \multicolumn{1}{c|}{\begin{minipage}[c][0.9cm][c]{1.1cm} \centering Width \end{minipage}} \\
        \hline
        TLCM & $2^{n \log n} O(m^2)$ & $O(\sqrt[3]{\rho \cdot n^2})$~\cite{10.1093/comjnl/bxad038} & $2^{O(\rho)}+n^{O(1)}$~\cite{DBLP:journals/ipl/KobayashiT16} & $\frac{\pi}{4}\sqrt{\frac{2^{n \log n}}{M}}$ & $O(m^2)$ & $O(n^2)$ & $O(m^2)$ \\
        \hline
        TLKP & $2^{n \log n} O(m^2)$ & $O(\sqrt{k} \cdot n)$~\cite{10.1093/comjnl/bxad038} & - & $\frac{\pi}{4}\sqrt{\frac{2^{n \log n}}{M}}$ & $O(m^2)$ & $O(m\log^2 m)$ & $O(m)$ \\
        \hline
        TLQP & $2^{n \log n} O(m^3)$ & $O(n)$~\cite{10.1093/comjnl/bxad038} & Para-\NP-hard~\cite{DBLP:conf/soda/AngeliniLBFP21} & $\frac{\pi}{4}\sqrt{\frac{2^{n \log n}}{M}}$ & $O(m^6)$ & $O(m^4)$ & $O(m^2)$ \\
        \hline    
        TLS & $O(m^\sigma n)$ & $O(n+\sigma)$ & $2^{O(\sigma^3)}n$~\cite{DBLP:journals/algorithmica/DujmovicFKLMNRRWW08} & $\frac{\pi}{4}\sqrt{\frac{2^{n \log n + \sigma \log m}}{M}}$ & $O(m^2)$ & $O(m)$ & $O(m)$ \\
        \hline    
        OPCM & $2^{n \log n} O(m^2)$ & $O(\sqrt[3]{\rho \cdot n^2})$~\cite{DBLP:conf/wg/ShahrokhiSSV94} & Courcelle's Th.~\cite{DBLP:journals/jgaa/BannisterE18} & $\frac{\pi}{4}\sqrt{\frac{2^{n \log n}}{M}}$ & $O(n^8)$ & $O(n^6)$ & $O(m^2)$ \\
        \hline    
        BT &  $2^{n \log n + m\log \tau} O(\tau n)$ & $O(\tau \cdot n)$ & Para-\NP-hard~\cite{Masuda87} & $\frac{\pi}{4}\sqrt{\frac{2^{n \log n + m \log \tau}}{M}}$ & $O(n^8)$ & $O(n^6)$ & $O(m)$ \\
        \hline    
        BS & $O(m^\sigma n)$ & $O(n+\sigma)$ & - & $\frac{\pi}{4}\sqrt{\frac{2^{n \log n + \sigma \log m}}{M}}$ & $O(n^8)$ & $O(n^6)$ & $O(m)$ \\
        \hline    
        \end{tabular}
        }
    \label{tab:circuit-model-results}
\end{table}

\paragraph{State of the art.}
We now provide an overview of the complexity status of each of the considered problems, together with the existence of FPT algorithms with respect to the corresponding natural parameter (total number of crossings $\rho$, number of crossings per edge $k$, maximum number of allowed mutually crossing edges, number of pages $\tau$, and number of edges to be removed $\sigma$), density bounds, and exact algorithms. Let $n$ and $m$ denote the number of vertices and edges of an input graph, respectively.

TLCM is probably the most studied among the above problems (see, e.g., \cite{DBLP:journals/algorithmica/DujmovicFKLMNRRWW08,DBLP:journals/jgaa/JungerM97}). It is \NP-complete~\cite{doi:10.1137/0604033}, and it remains \NP-complete even when the order on one level is prescribed~\cite{DBLP:journals/algorithmica/EadesW94}. %
Kobayashi and Tamaki~\cite{DBLP:journals/ipl/KobayashiT16} combined the kernelization result in~\cite{DBLP:journals/tcs/KobayashiMNT14} and an enumeration technique to device a fixed-parameter tractable (FPT) algorithm with running time in $2^{O(\rho)}+n^{O(1)}$.
Since the number of crossings may be quadratic in $m$, such an FPT result yields an algorithm whose running time is~$2^{O(m^2)}$.
On the other hand, a trivial $2^{n \log n} O(m^2)$-time exact algorithm for TLCM can be obtained by iteratively considering each of the possible $n! \in \Theta(2^{n \log n})$ vertex orderings, and by verifying whether the considered ordering yields less than $\rho$ crossings (which can be done in $O(m^2)$ time by considering each pair of edges). 
To the best of our knowledge, however, no faster exact algorithm is known for this problem that performs asymptotically better than the simple one mentioned above.
Observe that for positive instances of TLCM, $m$ is upper bounded by $\sqrt[3]{\frac{15625 \cdot \rho \cdot n^2}{4608}}$~\cite{10.1093/comjnl/bxad038}. 

TLKP has not been proved to be \NP-complete, and no FPT algorithm parameterized by $k$ is known for this problem. A trivial $2^{n \log n} O(m^2)$-time exact algorithm for TLKP can be devised analogously to the one for TLCM.
Observe that for positive instances of TLKP and for $k>5$, $m$ is upper bounded by $\frac{125}{96}\sqrt{k} \cdot n$~\cite{10.1093/comjnl/bxad038}.

TLQP is~\NP-complete \cite{DBLP:conf/soda/AngeliniLBFP21}. If we assume that its natural parameter is the maximum number of allowed mutually crossing edges, then no FPT algorithm exists for it parameterized by this parameter (unless \P $=$\NP). 
A trivial $2^{n \log n} O(m^3)$-time exact algorithm for TLQP can be devised analogously to the one for TLCM, where we have to test for the existence of crossings among triples of edges instead of pairs.
Observe that for positive instances of TLQP, $m$ is upper bounded by $2n - 4$~\cite{10.1093/comjnl/bxad038}.

TLS is \NP-complete~\cite{DBLP:journals/algorithmica/TanZ07}. %
Dujmovic et al. gave an FPT algortihm for TLS with $2^{O(\sigma^3)}n$ running time~\cite{DBLP:journals/algorithmica/DujmovicFKLMNRRWW08}.
A trivial $O(m^\sigma n)$-time exact algorithm for TLS performs a guess of $\sigma$ edges to be removed. This yields $\binom{m}{\sigma} \in O(m^\sigma)$ possible choices. For each of them, a linear-time algorithm to test if the input is a forest of caterpillars (and, thus, it admits a $2$-level planar drawing~\cite{DBLP:conf/acsc/EadesMW86}) is invoked.
Since caterpillars have at most $n-1$ edges, we have that for positive instances of TLS, $m$ is upper bounded by $n-1+\sigma$. 

OPCM is \NP-complete~\cite{Masuda87}. However, the optimum value of crossings can be approximated with an approximation ratio of $O(\log^2 n)$~\cite{DBLP:conf/wg/ShahrokhiSSV94}. Bannister and Eppstein~\cite{DBLP:journals/jgaa/BannisterE18} showed that OPCM is fixed-parameter tractable parameterized by $\rho$. %
To this aim, they exploit Courcelle’s Theorem~\cite{DBLP:journals/iandc/Courcelle90,DBLP:books/daglib/0030804}, which provides a super-exponential dependency of the running time in this parameter.
A trivial $2^{n \log n} O(m^2)$-time exact algorithm for OPCM can be devised analogous to the one for TLCM.
For positive instances of OPCM, $m$ is upper bounded by $\sqrt[3]{37 \cdot \rho \cdot n^2}$~\cite{DBLP:conf/wg/ShahrokhiSSV94}.

BT is \NP-complete, even when $\tau = 2$~\cite{Wig82}, in which case it coincides with the problem of testing whether the input graph is sub-Hamiltonian. This negative result implies that the problem does not admit FPT algorithms parameterized by $\tau$ (unless \P $=$\NP). 
A trivial $2^{n \log n + m \log \tau} O(\tau n)$-time exact algorithm for BT can be obtained by iteratively considering each of the possible choices of a permutation for the vertex order and of an assignment of the edges to the $\tau$ pages. This yields $n! \cdot \tau^m \in \Theta(2^{n \log n + m \log \tau})$ possible choices. For each of these choices, $\tau$ calls to a linear-time outer-planarity testing algorithm are performed, one for each of the graphs induced by the $\tau$ pages, to decide whether the considered choice defines a solution. 
Since outerplanar graphs have at most $2n-3$ edges and by the definition of BT, we have that for positive instances of BT, $m$ is upper bounded by $\tau (2n-3)$.

BS is \NP-complete~\cite{DBLP:journals/siamcomp/Yannakakis81} and no FPT algorithm parameterized by $\sigma$ is known for it.
A trivial $O(m^\sigma n)$-time exact algorithm for BS can be devised analogously to the one for TLS.
By the density of outerplanar graphs and by the definition of BS, we have that for positive instances of BS, $m$ is upper bounded by~$2n-3+\sigma$.

\section{Preliminaries}\label{se:preliminaries}

For basic concepts related to graphs and their drawings, we refer the reader, e.g.,  to~\cite{DBLP:books/ph/BattistaETT99,DBLP:reference/crc/2013gd}. For the standard notation we adopt to represent quantum gates and circuits, and for basic concepts about quantum computation, we refer the reader, e.g., to~\cite{DBLP:books/daglib/0046438,10.5555/1973124}. 

\paragraph{Notation.} Let $k$ be a positive integer. 
To ease the description, we will denote the value $\lceil \log_2  k \rceil$ simply as $\log k$, and the set $\{0,\dots,k-1\}$ as $[k]$. 
We refer to any of the permutations of the integers in $[k]$ as a \emph{$k$-permutation}. A \emph{$k$-set} is a set of size $k$.

We denote the set of binary values $\{0, 1\}$ by $\mathbb B$. Consider a binary string $s$ of length $a \cdot b$, for some $a, b \in \mathbb{N}$, i.e., $s \in \mathbb{B}^{a \cdot b}$. We often regard $s$ as a sequence of $a$ binary integers, each represented with $b$ bits (where the specific $a$ and $b$ will always be clarified in the considered context). For $i \in [a]$, the $i$-th number in $s$, which we denote by $s[i]$, is given by the substring of $s$ formed by the bits $s[b \cdot i][b \cdot  i + 1] \dots s[b \cdot i + b -1 ]$. 
Moreover, for $j \in [b]$, we denote by $s[i][j]$ the $j$-th digit of $s[i]$, where $s[i][0]$ is the least significant~bit~of~$s[i]$.

\paragraph{Graph drawing.} A \emph{drawing} of a graph maps each vertex to a point in the plane and each edge to a Jordan arc between its end-vertices. %
In this paper, we only consider graph drawings that are \emph{simple}, i.e., every two edges cross at most once and no edge crosses itself. 
A graph is \emph{planar} if it can be drawn in the plane such that no two edges cross, i.e., it admits a \emph{planar drawing}. A graph is \emph{$k$-planar} (with $k \geq 0$), if it can be drawn in the plane such that each edge is crossed at most $k$ times, i.e., it admits a \emph{$k$-planar drawing}. Finally, a graph is \emph{quasi-planar}, if it can be drawn in the plane so that no three edges pairwise cross, i.e., it admits a \emph{quasi-planar drawing}.  

Let $G$ be a graph. A \emph{$\tau$-page book layout} of $G$ consists of a linear ordering 
$\prec$ of the vertices of $G$ along a line, called the \emph{spine}, and of a partition $\{E_1,\dots,E_\tau\}$ of the edges of $G$ into $\tau$ sets, called \emph{pages}. A \emph{$\tau$-page book embedding} of $G$ is a $\tau$-page book layout such that no two edges of the same page \emph{cross}. That is, there exist no two edges $(u,v)$ and $(w,z)$ in the same page $E_i$ such that $u \prec v$, $v \prec w$, $u \prec z$, and $z \prec v$. The \emph{book thickness} of $G$ is the minimum integer $\tau$ for which $G$ has a $\tau$-page book embedding. The \emph{book skewness} of $G$ is the minimum number of edges that need to be removed from $G$ so that the resulting graph has book thickness $1$, that is, it is outerplanar~\cite{DBLP:journals/jct/BernhartK79}.

Let $G=(U,V,E)$ be a bipartite graph, where $U$ and $V$ denote the two subsets of the vertex set of $G$, and $E$ denotes the edge set of $G$. 
A \emph{$2$-level drawing} of $G$ maps each vertex $u \in U$ to a point on a horizontal line $\ell_u$, which we call the \emph{$u$-layer}, each vertex $v \in V$ to a point on a horizontal line $\ell_v$ (distinct from $\ell_u$), which we call the \emph{$v$-layer}, and each edge in $E$ to a $y$-monotone curve between its endpoints. 
Observe that, from a combinatorial standpoint, a $2$-level drawing $\Gamma$ of $G$ is completely specified by the linear ordering in which the vertices in $U$ and the vertices in $V$ appear along $\ell_u$ and $\ell_v$, respectively. The \emph{$2$-level skewness} of $G$ is the minimum number of edges that need to be removed from $G$ so that the resulting graph admits a $2$-level planar drawing, that is, it is a forest of caterpillars~\cite{DBLP:conf/acsc/EadesMW86}. 

\smallskip
\noindent Next, we provide the definitions of the search problems we study concerning $2$-level drawings of graphs. 
\vspace{-2mm}

\problemQuestion{\sc Two-level Crossing Minimization (TLCM)}%
{A bipartite graph $G$ and a positive integer~$\rho$.}%
{A $2$-level drawing of $G$ with at most $\rho$ crossings, if one exists.}

\problemQuestion{\sc Two-level $k$-planarity (TLKP)}%
{A bipartite graph $G$ and a positive integer $k$.}%
{A $2$-level $k$-planar drawing of $G$, if one exists.}

\problemQuestion{\sc Two-level Quasi Planarity (TLQP)}%
{A bipartite graph $G$.}%
{A $2$-level quasi-planar drawing of $G$, if one exists.}

\problemQuestion{\sc Two-level Skewness (TLS)}%
{A bipartite graph $G$ and an integer $\sigma$.}%
{A set $S \subseteq E(G)$ such that $|S| \leq \sigma$ and the graph $G' = (V,E(G) \setminus S)$ is a forest of caterpillars, if one exists.}

\smallskip
\noindent Finally, we provide the definitions of the search problems we study concerning book embeddings of graphs.

\problemQuestion{\sc One-Page Crossing Minimization (OPCM)}%
{A graph $G$ and a positive integer $\rho$.}%
{A $1$-page book layout of $G$ with at most $\rho$ crossings, if one exists.
}

\problemQuestion{\sc Book Thickness (BT)}%
{A graph $G$ and an integer~$\tau$.}%
{A $\tau$-page book embedding of $G$, if one exists.}

\problemQuestion{\sc Book Skewness (BS)}%
{A graph $G$ and an integer~$\sigma$.}%
{A set $S \subseteq E(G)$ such that $|S| \leq \sigma$ and the graph $G' = (V,E(G) \setminus S)$  is outerplanar, if one exists.}

\paragraph{Cross-independent sets.} 
Let $X$ be a ground set. 
Let $X_k$ be the set of all $k$-sets of distinct elements of $X$, i.e., $X_k = \{\{a_1,a_2,\dots,a_k\} | \forall i \neq j: a_i \neq a_j; a_i,a_j \in X\}$.
A subset $S$ of $X_k$ is \emph{cross-independent} if, for any two $k$-sets $s_i, s_j \in S$, it holds that $s_i \cap s_j = \emptyset$.
In order to prove the depth bounds of our circuits we will exploit the following.

\begin{lemma}\label{le:partion-k}
The set $X_k$ of all $k$-sets of distinct elements of a set $X$ can be partitioned in $O(\sqrt{k^{k}}\cdot |X|^{k-1})$ cross-independent sets of \mbox{size~at~most~$\lfloor \frac{|X|}{k} \rfloor$, if $k < \frac{|X|+2}{3}$.}
\end{lemma}

\begin{proof}
The fact that a cross-independent subset $S$ of $X_k$ contains at most $\lfloor \frac{|X|}{k} \rfloor$ elements is trivial, since each element of $S$ is a $k$-set and no two elements of $S$ may contain the same element of $X$. To show that $X_k$ admits a partition into $O(\sqrt{k^{k}}\cdot|X|^{k-1})$ cross-independent sets we proceed as follows. Consider the auxiliary graph ${\cal A}(X_k)$, whose vertices are in 1-to-1 correspondence with the elements of $X_k$, i.e., the $k$-sets of distinct elements of $X$. The edges of 
${\cal A}(X_k)$ connect pairs of vertices corresponding to $k$-sets with a non-empty intersection. We show that the maximum degree of ${\cal A}(X_k)$ is $O(\sqrt{k^{k}}\cdot|X|^{k-1})$. This immediately implies that ${\cal A}(X_k)$ admits a proper coloring with $O(\sqrt{k^{k}}\cdot|X|^{k-1})$ colors. Since each color class induced an independent set in ${\cal A}(X_k)$, we have that the vertices of ${\cal A}(X_k)$ in the same color class form a cross-independent set, which yields the proof that  $O(\sqrt{k^{k}}\cdot|X|^{k-1})$ cross-independent sets suffice to partition $X_k$.

We denote each vertex $v$ of ${\cal A}(X_k)$ by the corresponding $k$-set $\{a_1,a_2,\dots,a_k\}$ in $X_k$. Recall that $v$ is adjacent to all the vertices of ${\cal A}(X_k)$ whose corresponding sets contain at least one of $a_1,a_2,\dots,a_k$. We have that the number of vertices of  ${\cal A}(X_k)$ that contain a subset of $\{a_1,a_2,\dots,a_k\}$ of size $r$ is $\binom{k}{r}\binom{|X|-k}{k-r}$. Therefore, the degree of $v$ is upper bounded by 

\begin{equation}\label{eq:cross-independent}
\sum^{k-1}_{i=1}\binom{k}{i}\binom{|X|-k}{k-i}
\end{equation}

Note that, the term $t_1(i) = \binom{k}{i}$ is maximum when $i = \frac{k}{2}$, i.e., when $t_1(\frac{k}{2}) = \binom{k}{\frac{k}{2}} \in O(\sqrt{k^k})$. On the other hand, the term $t_2(i) = \binom{|X|-k}{k-i} \in O(|X|^{k-i})$
is monotone as long as $k-i < \frac{|X|-k}{2}$, i.e., $k < \frac{|X|+2i}{3}$. Therefore, since $1 \leq i \leq k-1$, we get that $t_2(i)$ is monotone under our hypothesis that $k < \frac{|X|+2}{3}$. Therefore, if $k < \frac{|X|+2}{3}$,
we get that \cref{eq:cross-independent}, which can we rewritten as 
$\sum^{k-1}_{i=1} t_1(i) t_2(i)$, is upper bounded by 
$t_1(\frac{k}{2}) \sum^{k-1}_{i=1} |X|^{k-i} \leq \binom{k}{\frac{k}{2}} \cdot 2 |X|^{k-1}$. This shows the claimed bound.
\end{proof}

We remark that, for $k=2$, the asymptotic bounds of \cref{le:partion-k} can also be derived from Vizing's Theorem~\cite{Vizing64}.

\paragraph{Mathematical formulations.} 
We introduce the mathematical formulations used in the the D-Wave quantum annealing platform.

A \emph{constrained binary optimization} (CBO) is the mathematical formulation of an optimization problem, in which the variables are binary. Note that, both the objective function and the constraints may have an arbitrary degree. In some cases, we focus on CBO formulations in which the objective function is not defined, and we aim at verifying whether a problem instance satisfies the given constraints.

A \emph{quadratic unconstrained binary optimization} (QUBO) is the mathematical formulation of an optimization problem, in which the variables are binary, the optimization function is quadratic, and there are no constraints. Specifically, 
let $Q$ be an upper triangular matrix $Q \in {\mathbb R}^{k\times k}$. Using $Q$, we can define a quadratic function $f_Q: {\mathbb B}^{k} \rightarrow \mathbb{R}$ that assigns a real value to a $k$-length binary vector. Namely, we let $f_Q(x) = x^TQx = \sum^k_{i=1}\sum^k_{j=1} Q_{ij}x_ix_j$.
The QUBO formulation for $f_Q$ asks for the binary vector $x^*$ that minimizes $f_Q$, i.e.,  $x^*= \underset{x \in {\mathbb B}^k}\argmin f_Q(x)$.

\section{Basic Quantum Circuits} \label{se:basic-quantum}

In this section, we introduce basic quantum circuit which we will exploit in~\cref{se:input-initializers,se:solution-detectors}. 
Let $D$ be a vertex-weighted directed acyclic graph (DAG). 
The \emph{depth} of $D$ is the number of vertices in a longest path of $D$. 
Two vertices $u$ and $v$ of $D$ are \emph{incomparable} if there exists no directed path from $u$ to $v$, or vice versa. 
An \emph{anti-chain} of $D$ is a maximal set of incomparable vertices. 
The \emph{weight} of an anti-chain is the sum of its vertex weights. 
The \emph{width} of $D$ is the maximum weight of an anti-chain of $D$. 
A quantum circuit $Q$ can be modeled as
a vertex-weighted directed acyclic graph $D_Q$, whose vertices correspond to the gates of $Q$ and whose directed edges represent
qubit input-output dependencies. Moreover, the weight of a vertex representing a gate $U$ corresponds to the number of elementary gates~\cite{DBLP:books/daglib/0046438,10.5555/1973124} needed to build~$U$.

The \emph{circuit complexity} of a quantum circuit is the number of elementary gates used to construct it.
The \emph{depth} and the \emph{width} of a quantum circuit are the depth and the width, respectively, of its associated weighted DAG.
Note that, 
the size of a circuit corresponds to the total number of operations that must be performed to execute the circuit,
the depth of a circuit corresponds to the number of distinct time steps at which gates are applied, and %
the width of a circuit corresponds to the maximum number of operations that can be performed ``in parallel''. 
Therefore, it is natural to upper bound the time complexity either by its depth, assuming the gates at each layer can be executed in parallel, or by its circuit complexity, assuming the gates are executed sequentially. In the lemmas and theorems that will follow, we describe circuits in terms of their circuit complexity, depth, and width. The width of the circuit is reported as an indication of the desired level of parallelism.

We denote by $\ket{0_k}$ the quantum basis state composed of $k$ qubits set to $\ket{0}$. 
We now describe some gates that will be used in the following sections. %
Let $\phi[i]$ and $\phi[j]$ be binary strings of length $\log t$, which we interpret as binary integers represented with $\log t$ bits. Also, let $\ket{\phi[i]}$ and $\ket{\phi[j]}$ be the basis states corresponding to $\phi[i]$ and $\phi[j]$, respectively.  First, we focus on gate $U_{=}$ that, given integers $\phi[i]$ and $\phi[j]$, verifies if $\phi[i]$ is equal to $\phi[j]$.

\begin{lemma}\label{le:gate-equal}
There exists a gate $U_{=}$ that, when provided with the input superposition $\ket{\phi[i]}\ket{\phi[j]}\ket{0_{\log t}}\ket{0}$, produces the output superposition $\ket{\phi[i]}\ket{\phi[j]}\ket{0_{\log t}}\ket{\phi[i]=\phi[j]}$. Gate $U_{=}$ has 
$O(\log t)$ circuit complexity, depth, and width.
\end{lemma}

\begin{proof}
\begin{figure}[tb!]
    \centering
    \includegraphics[page = 2, width = .8\textwidth]{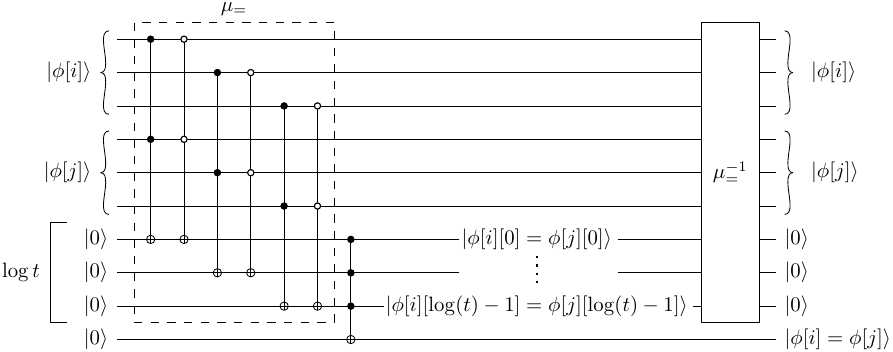}
    \caption{The gate $U_{=}$.}
    \label{fig:equal}
\end{figure}
Gate $U_{=}$ consists of two gates $\mu_{=}$ and $\mu_{=}^{-1}$ and in between such gates it executes a Toffoli gate; refer to \cref{fig:equal}. 
The input of $\mu_{=}$ is the superposition $\ket{\phi[i]}\ket{\phi[j]}\ket{0_{\log t}}.$ The output of $\mu_{=}$ is the superposition $$\ket{\phi[i]}\ket{\phi[j]}\ket{\phi[i][0] = \phi[j][0]}\dots \ket{\phi[i][\log(t)-1] = \phi[j][\log(t)-1]}.$$ 
Gate $\mu_{=}$, for each $a \in [\log t]$, computes qubit $\ket{\phi[i][a] = \phi[j][a]}$ with two Toffoli gates with three inputs and outputs. The input to both Toffoli gates are the two control qubits $\ket{\phi[i][a]}$, $\ket{\phi[j][a]}$, and a target qubit initialized to $\ket{0}$. 
The first Toffoli gate is activated when $\phi[i][a]=\phi[j][a]=1$. The second Toffoli gate is activated when $\phi[i][a]=\phi[j][a]=0$. The target qubit is set to $\ket{0 \oplus (\phi[i][a] = \phi[j][a])}$.
Qubits $\ket{\phi[i][0] = \phi[j][0]}\dots \ket{\phi[i][\log(t) -1] = \phi[j][\log(t) -1]}$ form the input of a Toffoli gate with $\log t +1$ inputs and outputs, which computes qubit $\ket{\phi[i]=\phi[j]}$. The control qubits are $\ket{\phi[i][0] = \phi[j][0]}\dots \ket{\phi[i][\log(t) -1] = \phi[j][\log(t) -1]}$. The target qubit is initialized to $\ket{0}$. The Toffoli gate is activated if $\ket{\phi[i][a] = \phi[j][a]}$ is equal to $\ket{1}$, for all $a \in [\log t]$. The target qubit is set to $\ket{0 \oplus \bigwedge_{a \in [\log t]}(\phi[i][a] = \phi[j][a])}$.
The qubits $$\ket{\phi[i]}\ket{\phi[j]}\ket{\phi[i][0] = \phi[j][0]}\dots \ket{\phi[i][\log(t) -1] = \phi[j][\log(t) -1]}$$ then enter $\mu_{=}^{-1}$ that, being the inverse of $\mu_{=}$, outputs the superposition $\ket{\phi[i]}\ket{\phi[j]}\ket{0_{\log t}}$.
Overall, $U_{=}$ is implemented using $4 \log t$ Toffoli gates with a constant number of inputs and outputs and one Toffoli gate with $\log t + 1$ inputs and outputs. In turn, this last Toffoli gate is implemented using $\log t$ Toffoli gates with $O(1)$ inputs and outputs. Gate $U_=$ has circuit complexity $O(\log t)$. Hence, it has the same bound for its depth and width.
\end{proof}

Second, we focus on gate $U_{<}$  that, given binary integers $\phi[i]$ and $\phi[j]$ represented with $\log t$ bits, verifies if $\phi[i]$ is less than $\phi[j]$.

\begin{lemma}\label{le:gate-less}
There exists a gate $U_{<}$ that,  when provided with the input superposition $\ket{\phi[i]}\ket{\phi[j]}\ket{0_{\log{t}}}\ket{0}$, produces the output superposition $\ket{\phi[i]}\ket{\phi[j]}\ket{0_{\log{t}}}\ket{\phi[i]<\phi[j]}$. Gate $U_{<}$ has 
$O(\log t)$ circuit complexity, depth, and width.
\end{lemma}

\begin{proof}
\begin{figure}[tb!]
    \centering
    \includegraphics[page = 8, width = .8\textwidth]{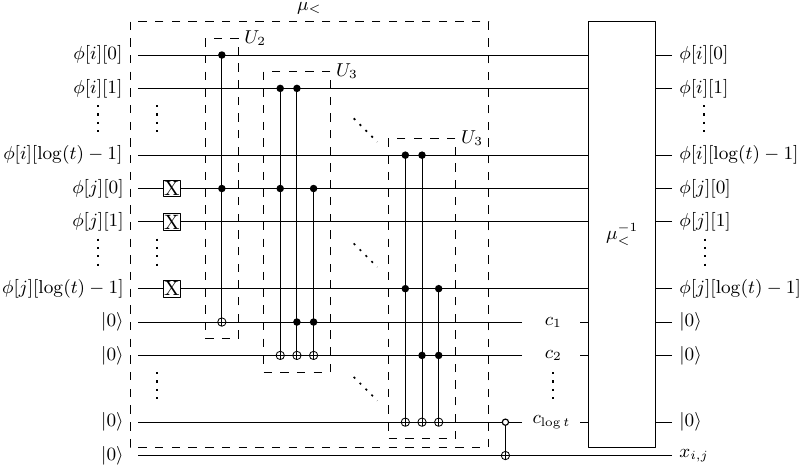}
    \caption{The gate $U_<$.}
    \label{fig:less}
\end{figure}
Gate $U_{<}$ consists of two gates $\mu_{<}$ and $\mu_{<}^{-1}$ and in between such gates it executes an Anticontrolled NOT gate; refer to \cref{fig:less}. 
The input of $\mu_{<}$ is the superposition $\ket{\phi[i]}\ket{\phi[j]}\ket{0_{\log t}}.$ The output of $\mu_{<}$ is the superposition $\ket{\phi[i]}\ket{\phi[j]}\ket{c_0}\ket{c_2} \dots \ket{c_{\log (t)-1}}$, where $c_0$ is the carry of the sum  $\phi[i][0] + \phi[j][0]$ and $c_k$ is equal to the carry of the sum of $c_{k-1} + \phi[i][k] + \phi[j][k]$, with $2 \leq k < \log n$. 

Gate $\mu_{<}$ uses two gates $U_2$ and $U_3$. Gate $U_2$ computes the carry of the sum of two qubits, receiving as input $\ket{a}\ket{b}\ket{0}$ and outputs $\ket{a}\ket{b}\ket{a \wedge b}$. Gate $U_3$ computes the carry of the sum of three qubits, receiving as input $\ket{a}\ket{b}\ket{c}\ket{0}$ and outputs $\ket{a}\ket{b}\ket{c}\ket{(a \wedge b)\vee(a \wedge c)\vee(b \wedge c)}$. 
Gate $\mu_{<}$ first computes the complement to one $\overline{\phi[j]}$ of $\phi[j]$, using bit-flip (Pauli-$X$) gates, then it performs $U_2$ between $\phi[i][0]$ and $\overline{\phi[j][0]}$. The carry qubit of the previous sum, together with $\phi[i][1]$ and $\overline{\phi[j][1]}$ enters $U_3$. 
For $k=2, \dots \log (t) -1$ it performs the $U_3$ gate between the carry qubit of the prevoius sum, and $\phi[i][k]$ and $\overline{\phi[j][k]}$.
Note that, the last carry qubit is equal to one if and only if $\phi[i]$ is greater than $\phi[j]$. 

The last carry qubit of $\mu_{<}$ is used as the input control qubit of a {\sc Anticontrolled NOT} gate whose input target qubit is set to $\ket{0}$. The output target qubit is therefore ``flipped'' only if the input control qubit is $\ket{0}$, i.e., if $\phi[i]<\phi[j]$. 

All qubits, except for the one computed above, now enter gate $\mu_{<}^{-1}$ that is the inverse gate of $\mu_{<}$.

Overall, $U_{<}$ is implemented using $2 \log t$ Pauli-$X$ gates with one input and one output, and $6 \log t +2$ Toffoli gates with three inputs and outputs. Plus, it uses an Anticontrolled NOT, with two inputs and outputs.
This shows that the circuit complexity of $U_{<}$ is $O(\log t)$. Hence, it has the same bound for the depth and for the width.
\end{proof}

Third, we focus on gate $U_{1s}$  that, given a binary string $b$ of length $t$, counts how many bits set to $1$ are in it. For simplicity, we assume that $t$ is a power of $2$. If not, we can always append to $b$ the smallest number of $0$s such that this property holds. Observe that, in the worst case, the length of the string may double, which does not alter the asymptotic bounds of the next lemma.

\begin{lemma}\label{le:gate-counter}
There exists a gate $U_{1s}$ that, when provided with the input superposition $\ket{b_0 \dots b_{t-1}}\ket{0_{h}}\ket{0_{k}}$, where $t$ is a power of $2$, $k = 1 + \log t$, and $h = 4t-2\log t -4 $, produces the output superposition $\ket{b_0 \dots b_{t-1}}\ket{0_{h}}\ket{s}$, where $s$ is the binary representation, in $\log t + 1$ bits, of the total number $\sum^{t-1}_{i=0} b_i$ of qubits set to $1$ in $b$. Gate $U_{1s}$ has $O(t)$ circuit complexity, $O(\log^2 t)$ depth, and $O(t)$ width.
\end{lemma}

\begin{proof}
\begin{figure}[tb!]      
\centering
\includegraphics[page = 38,width=.6\textwidth]{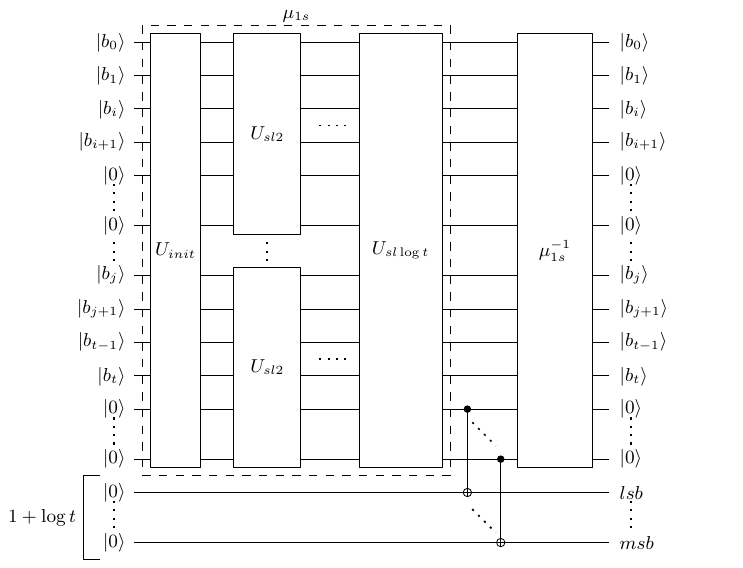}
\caption{Structure of gate $U_{1s}$.}
\label{fig:U_s1}
\end{figure}

\begin{figure}
\centering
\includegraphics[page = 58,width=.75\textwidth]{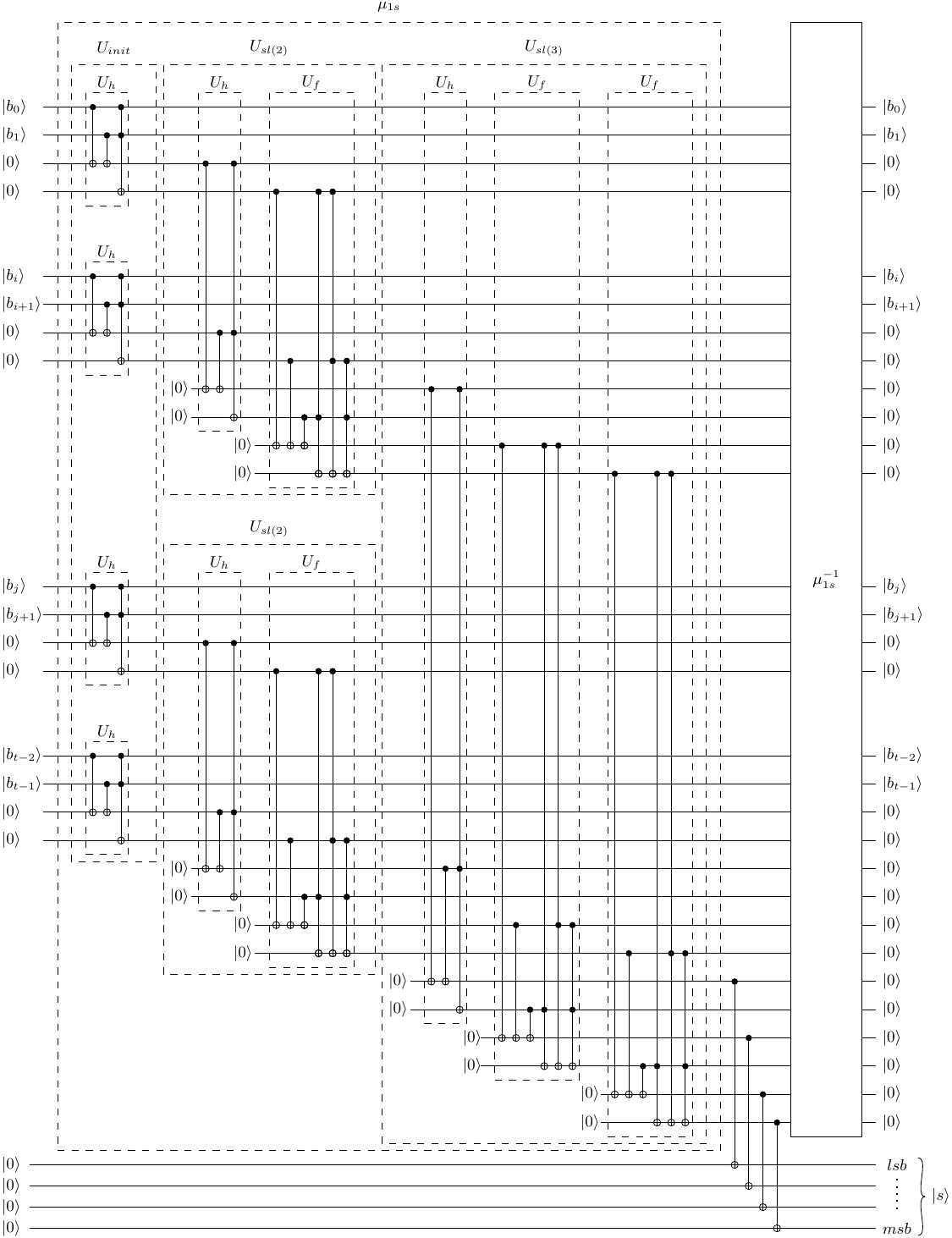}
\caption{First three steps of gate $U_{1s}$, with $t = 8$.}
\label{fig:U_part_s1}
\end{figure}

Gate $U_{1s}$ consists of two gates $\mu_{1s}$ and $\mu_{1s}^{-1}$, and in between such gates it executes $\log(t) +1$ C-NOT gates; refer to~\cref{fig:U_s1}. 
Let $b = b_0 \dots b_{t-1}$, i.e., $b_i$ is the $i$-th bit of the binary string $b$. 
The input of $\mu_{1s}$ is the superposition $\ket{b}\ket{0_{h-1-\log t}}\ket{0_{1+\log t}}$. Gate $\mu_{1s}$ outputs the superposition $\ket{b}\ket{0_{h-1-\log t}}\ket{s}$, where $s=\sum^{t-1}_{i=0} b_i$ is the binary representation, in $\log t +1$ qubits, of the total number of bits set to $1$ in $b$.

Gate $\mu_{1s}$ exploits the gate $U_{init}$ and the gates $U_{sl(i)}$, for $i = 2,\dots,\log(t)$; refer to~\cref{fig:U_part_s1}.
Specifically, $\mu_{1s}$ first executes the gate $U_{init}$. Then, it executes $\frac{t}{4}$ parallel gates $U_{sl(2)}$. Then, for $i=3,\dots,\log(t)$, in gate $\mu_{1s}$, the $\frac{t}{2^{i-1}}$  parallel gates $U_{sl({i-1})}$ are followed by $\frac{t}{2^{i}}$  parallel gates $U_{sl({i})}$; refer to~\cref{fig:U_part_s1}.

\begin{figure}[tb!]
        \begin{subfigure}[b]{0.45\textwidth}    
        \centering
        \includegraphics[page = 4]{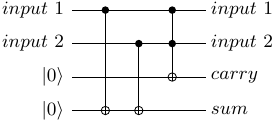}
        \caption{Gate $U_h$}
        \label{fig:half-adder}
        \end{subfigure}
        \begin{subfigure}[b]{0.45\textwidth}
        \centering
        \includegraphics[page = 6]{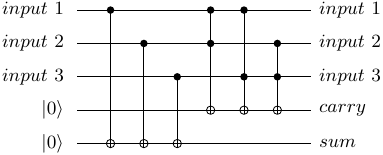}
        \caption{Gate $U_f$}
        \label{fig:full-adder}
        \end{subfigure}
        \caption{The half-adder gate $U_h$ (left) and the full-adder gate $U_f$ (right).}
        \label{fig:half-full}
\end{figure}
The purpose of gate $U_{init}$ is to ``transform'' the binary string $b$ of length $t$ in $\frac{t}{2}$ binary strings $bn_i$ each of length $2$, such that $bn_i$ is the binary encoding of the number $b_{2i} + b_{2i+1}$, for $i = 0,\dots,\frac{t-2}{2}$. Namely, this gate partitions the bits of $b$ into pairs and sums each pair to form a binary integer represented using two bits.
Gate $U_{init}$ exploits $\frac{t}{2}$ parallel gates {\sc Half-Adder} ({\sc HA}) $U_h$, which we describe next; refer to \cref{fig:half-adder}. 
Each gate HA takes in input the superposition $\ket{a b}\ket{0}\ket {0}$, where $a,b \in {\mathbb B}$, and outputs the superposition $\ket{a b}\ket{z}\ket{c}$, where $z = a \oplus b$ is the least significant bit of the binary sum $a + b$ and $c= a \wedge b$ is the most significant bit of the binary sum $a + b$. Observe that, $c$ is the carry bit of the bitwise sum of $a$ and $b$.
For each $i=0,\dots,\frac{t-2}{2}$, the $U_{init}$ exploits a gate HA to which the bits $b_{2i}$ and $b_{2i+1}$ are provided in place of the bits $a$ and $b$, respectively; see~\cref{fig:U_part_s1}. Note that, each gate HA has $O(1)$ circuit complexity, depth, and width. Thus, the gate $U_{init}$ has $O(t)$ circuit complexity, $O(1)$ depth, and $O(t)$ width.

The purpose of the gate $U_{sl(i)}$ is to compute the sum of two binary integers of length $i$, which it then outputs as a binary integer of length $i+1$. 
This gate exploit $1$ gate HA and $i-1$ gates {\sc Full-Adder} ({\sc FA}) $U_f$, which we describe next; refer to \cref{fig:full-adder}.
Each gate FA takes in input the superposition $\ket{a b c}\ket{0}\ket {0}$, where $a, b,c \in {\mathbb B}$, and outputs the superposition $\ket{a b c}\ket{z}\ket{c}$, where $z = a \oplus b \oplus c$ is the least significant bit of the binary sum $a+b+c$ and $c= (a \wedge b) \vee (a \wedge c) \vee (b \wedge c)$ is the most significant bit of the binary sum $a + b + c$. Observe that, $c$ is the carry bit of the bitwise sum of $a$, $b$, and $c$.
The gate $U_{sl(i)}$ uses a gate HA to compute the sum of the two least significant bits of its two input binary integers; let $c_0$ be the corresponding carry bit (possibly $c_0 = 0$). Then, it uses a gate FA to compute the sum of the carry bit $c_0$ with the second-least significant bit of the two input binary integers; let $c_1$ be the corresponding carry bit. Then, for $k=1,\dots,i-1$, it uses a gate FA to compute the sum of the carry bit $c_k$ with the $(k+1)$-th least-significant bit of the two input binary integers. Note that, each gate FA has $O(1)$ circuit complexity, depth, and width. Thus, the gate $U_{sl(i)}$ has $O(i)$ circuit complexity, $O(i)$ depth, and $O(1)$ width.
Altogether, the gate $U_{sl(i)}$ takes in input the superposition $\ket{a_0\dots a_{i-1}}\ket{b_0\dots b_{i-1}}\ket{0_{2i}}$, where $a_0,\dots,a_{i-1}, b_0,\dots,b_{i-1} \in {\mathbb B}$, and outputs the superposition $\ket{a_0 \dots a_{i-1}}\ket{b_0\dots b_{i-1}}\ket{z}\ket{c}\ket{g}$, where $z \in {\mathbb B}^i$, $c=c_i \in {\mathbb B}$, and the binary string $cz$ coincides with  
$a_0a_1\dots a_{i-1} + b_0b_1\dots b_{i-1}$ and $g=c_0 c_1 \dots c_{i-1}$ is the concatenation of the carry bits of the gate HA and of the gates FA.
Since the gate $U_{sl(i)}$ consists of one gate HA and $i-1$ gates FA, it has $O(i)$ circuit complexity, $O(i)$ depth, and $O(1)$ width.

To prove some of the next bounds, we will exploit the following.

\begin{claimx}\label{cl:series}   
For any positive integer $k \geq 1$, let $\phi(k) = \sum^{k}_{i=1} \frac{i}{2^i}$. It holds that 
\begin{equation}\label{eq:sum-log-minus-1}
\phi(k) = \frac{1}{2^k}(2^{k+1}-k-2)
\end{equation} 
\end{claimx}

\begin{proof}
We prove the statement by induction on $k$.

In the base case $k=1$. Then, by definition, $\phi(1)= \sum^{1}_{i=1} \frac{i}{2^i} = \frac{1}{2}$. 
Also, by \cref{eq:sum-log-minus-1}, $\phi(1)=\frac{1}{2}(4-1-2) = \frac{1}{2}$. Therefore, the statement holds.

Suppose now that $k>1$. Then, by definition, we have
\begin{equation}\label{eq:minus-1}
\phi(k) = \sum^{k}_{i=1} \frac{i}{2^i} = \sum^{k-1}_{i=1} \frac{i}{2^i} + \frac{k}{2^{k}} = \phi(k-1) + \frac{k}{2^{k}}.
\end{equation}

By induction and by \cref{eq:sum-log-minus-1}, we have that  $\phi(k-1) = \frac{1}{2^{k-1}}(2^{k}-k-1)$.  Thus, by replacing $\phi(k-1)$ with such an expression in \cref{eq:minus-1}, we have

\begin{multline}\label{eq:simple-induction}
\phi(k) = \frac{1}{2^{k-1}}(2^{k}-k-1)  + \frac{k}{2^{k}} = \\ \frac{1}{2^k}(2(2^k -k-1) + k) = 
\frac{1}{2^k}(2^{k+1} -2k-2 + k) = \frac{1}{2^k}(2^{k+1} -k-2).
\end{multline}

\cref{eq:simple-induction} concludes the proof of the inductive case and of the claim.
\end{proof}

Since gate $\mu_{1s}$ contains a gate $U_{init}$ and the gates $U_{sl(i)}$, for $i = 2,\dots,\log(t)$, we have the following. The circuit complexity of $\mu_{1s}$ 
can be estimated as follows. Gate $U_{init}$ contains $\frac{t}{2}$ gates HA. Also, for $i = 2,\dots,\log(t)$, gate $\mu_{1s}$ contains $\frac{t}{2^i}$ gates $U_{sl(i)}$.
Moreover, each gate $U_{sl(i)}$ contains a total number of gate HA and FA equal to $i$.
Therefore, the overall number of gates HA and FA in $\mu_{1s}$ is $Y = t \cdot \sum^{\log t}_{i=1} \frac{i}{2^i}$. Since $Y = t \cdot \phi(\log t)$,  by \cref{cl:series} we have that $Y = 2t-\log(t)-2$, which is in $O(t)$. 
 Therefore, the circuit complexity of $\mu_{1s}$ is upper bounded by $O(t)$. The depth of $\mu_{1s}$ in $O(\sum^{\log t}_{i=1} i) \in O(\log^2 t)$. Finally, the width of $\mu_{1s}$ is bounded by the number of parallel circuits in $U_{init}$, which is~$O(t)$.

Observe that, the circuit $U_{sl(\log t)}$ is the last circuit of $\mu_{1s}$, and its output (which coincides with the one of $\mu_{1s}$) is the superposition $\ket{b}\ket{0_{h-1-\log t}}\ket{s}$, where $s=\sum^t_{i=1} b_i$ is the binary representation, in $\log t +1$ qubits, of the total number of bits set to $1$ in $b$. In order to allow the reuse of the ancilla qubit of $\mu_{1s}$, and still include the  $\log t +1$ qubits of $s$ in the output of $U_{1s}$, the gate $U_{1s}$ contains $\log t +1$ C-NOT gates, which are then followed by the inverse circuit $\mu^-1_{1s}$. In particular, the control qubit of each of these C-NOT gates is one of the  $\log t +1$ bits of $s$ and the target qubit is initialized to $\ket{0}$. Since C-NOT gates have $O(1)$ circuit complexity, we have that the circuit complexity, depth, and width of the gate $U_{1s}$ have the same asymptotic bounds of the circuit complexity, depth, and width of the gate $\mu_{1s}$.

To complete the proof, we show the bounds on the number of ancilla qubit of the gate $U_{1s}$. First, observe that each C-NOT gate uses exactly one ancilla qubit, which it turned into one of the bits encoding the sum  $\sum^t_{i=1} b_i$
Since $1 + \log t$ bits may be needed to encode such a sum, we use $1 + \log t$ C-NOT gates. This shows that $k = 1 + \log t$. For the value $h$, observe that each gate HA and FA used in the gates $\mu_{1s}$ and  $\mu^{-1}_{1s}$ exploits two ancilla qubits. Therefore, the number $h$ of ancilla qubits in input to $\mu_{1s}$ (and to $\mu^{-1}_{1s}$) is $2Y = 2 (2t-\log(t)-2) = 4t - 2\log(t) - 4$. This concludes the proof.
\end{proof}

\section{A Quantum Framework for Graph Drawing Problems}

In this section, we establish a framework for dealing with several \NP-complete graph drawing problems; refer to~\cref{fig:framework}. The framework is based on the Grover's approach for quantum search~\cite{DBLP:conf/stoc/Grover96}, which builds upon three circuits. The first circuit is a Hadamard gate that builds a uniform superposition of a sequence of qubits representing a potential, possibly not well-formed, solution to the problem. The second circuit exploits an {\em oracle} to perform the so-called \emph{\sc Phase Inversion}. The third circuit executes the so-called \emph{\sc Inversion about the Mean}. The second and the third circuit are executed a number of times which guarantees that a final measure outputs a solution, if any, with high probability.

\begin{figure}[tb!]
    \centering
    \includegraphics[page = 48, width = .85\textwidth]{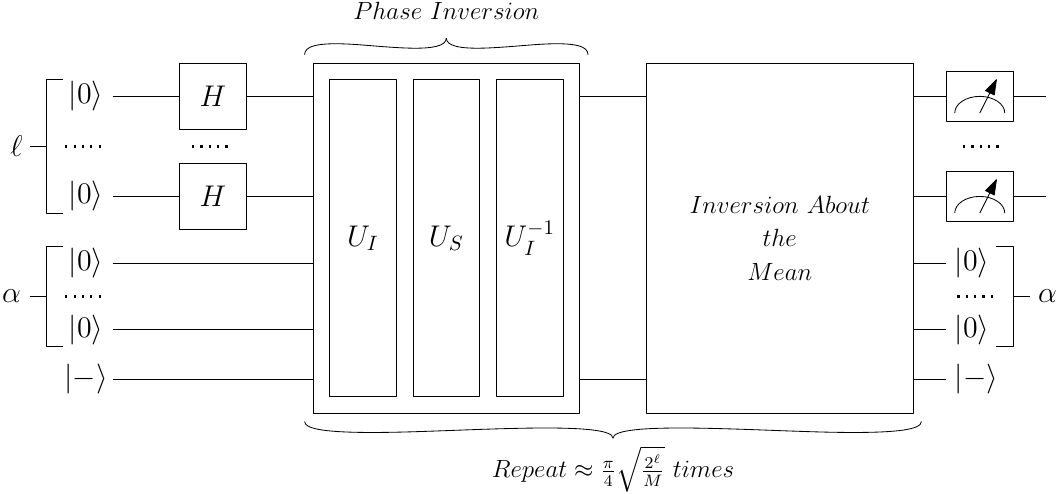}
    \caption{Overview of the quantum graph drawing framework based on Grover's approach.}
    \label{fig:framework}
\end{figure}

\begin{theorem}[Grover's search~\cite{GroverAaronson,DBLP:conf/stoc/Grover96}]\label{th:grover}
Let $P$ be a search problem whose solutions can be represented using $\ell$ bits and suppose that there exists a {\sc Phase Inversion} circuit for $P$ with $c(\ell)$ circuit complexity and $d(\ell)$ depth. Assume that $c(\ell)$ and $d(\ell)$ are $\Omega(\log \ell)$.
Then, there exists a quantum circuit that outputs a solution for $P$, if any, with $\frac{\pi}{4} \sqrt{\frac{2^\ell}{M}} \cdot O( c(\ell))$ circuit complexity and $\frac{\pi}{4} \sqrt{\frac{2^\ell}{M}} \cdot O( d(\ell))$ depth, where $M$ is the number of solutions of $P$.
\end{theorem}

Let $n$ and $m$ be the number of vertices and edges of an input~graph~$G$, respectively. Observe that, in all the problems we consider, $G$ admits the sought layout if and only if each of its connected components does. Hence, in the following, we assume that $G$ is connected, and therefore $m \geq n-1$.
During the computation, we will manage a superposition $\ket{\Gamma} = \ket{\Phi}\ket{\Psi}\ket{\Theta}$, where $\ket{\Phi}$ is a superposition of $n \log n$ qubits, $\ket{\Psi}$ is a superposition of $m \log \tau$ qubits, and $\ket{\Theta}$ is a superposition of $\sigma \log m$ qubits. In particular, for some of the problems, $\ket{\Psi}$ and/or $\ket{\Theta}$ might not be present. We denote by $\ell$ the value $(n \log n) + (m \log \tau) + (\sigma \log m)$, where the second and/or third terms might be missing. 

The superposition $\ket{\Phi} = \sum_{\phi \in {\mathbb B}^{n \log n}} c_\phi \ket{\phi}$ represents the superposition of all sequences of $n$ natural numbers with values in $[n]$, each represented by a binary string $\phi$ of length $n\log n$ (according to notation defined in ~\cref{se:preliminaries} to represent sequences of integers). The digits corresponding to each natural number contained in $\phi$ form a consecutive sequence of length $\log n$. 
In particular, we denote by $\phi[i]$ the $i$-th natural number contained in $\phi$. %
The purpose of $\ket{\phi}$ is to represent the position of each vertex of $G$ in a total order. To this aim, observe that, within the superposition $\ket{\Phi}$, all possible states corresponding to assignments of positions from $0$ to $n-1$ for each vertex in $G$ are included.

The superposition $\ket{\Psi} = \sum_{\psi \in {\mathbb B}^{m \log\tau}} c_\psi \ket{\psi}$ represents the superposition of all sequences of $m$ natural numbers with values in $[\tau]$, each represented by a binary string $\psi$ of length $m\log \tau$. The digits corresponding to each natural number contained in $\psi$ form a consecutive sequence of length $\log \tau$. In particular, we denote by $\psi[i]$ the $i$-th natural number contained in $\psi$. 
The purpose of $\ket{\psi}$ is to represent a coloring of the edges of $G$ with color in $[\tau]$. To this aim, observe that, within the superposition $\ket{\Psi}$, all possible states corresponding to an assignment of integers from $0$ to $\tau-1$ for each edge in $G$ are included.

The superposition $\ket{\Theta} = \sum_{\theta \in {\mathbb B}^{\sigma \log m}} c_\theta \ket{\theta}$ represents the superposition of all sequences of $\sigma$ natural numbers with values in $[m]$, each represented by a binary string $\theta$ of length $\sigma \log m$. The digits corresponding to each natural number contained in $\theta$ form a consecutive sequence of length $\log m$. In particular, we denote by $\theta[i]$ the $i$-th natural number contained in $\theta$. 
The purpose of $\ket{\theta}$ is to represent a subset of the edges of $G$ of size at most $\sigma$, each labeled with an integer in $[m]$. To this aim, observe that, within the superposition $\ket{\Theta}$, all possible states corresponding to a selection of $\sigma$ edges of $G$, where each edge is indexed with an integer from $0$ to $m-1$, are included.

For problems {\sc TLCM}, {\sc TLKP}, {\sc TLQP}, and {\sc OPCM}, we have that $\ket{\Gamma} = \ket{\Phi}$.
For problem {\sc BT}, we have that $\ket{\Gamma} = \ket{\Phi}\ket{\Psi}$.
Finally, for problems {\sc TLS} and {\sc BS}, we have that $\ket{\Gamma} = \ket{\Phi}\ket{\Theta}$.

Next, we present an overview of how the superposition $\ket{\Gamma}$ evolves within the three main circuits of the framework; refer to \cref{fig:framework}.

First, in all problems we study, $\ell$ qubits set to $\ket{0}$ enter an Hadamard gate that outputs the uniform superposition $\ket{\Gamma} = H^{\otimes \ell}\ket{0_{\ell}} = \frac{1}{\sqrt{2^{\ell}}}\sum_{\gamma \in {\mathbb B}^{\ell}} \ket{\gamma}$. Observe that, such a superposition, corresponds to the tensor product of the uniform superpositions $\ket{\Phi} = H^{\otimes {n \log n}}\ket{0_{n \log n}}$, $\ket{\Psi} = H^{\otimes {m \log \tau}}\ket{0_{m \log \tau}}$, and $\ket{\Theta} = H^{\otimes {\sigma \log m}}\ket{0_{\sigma \log m}}$, where possibly $\Psi$ and/or $\Theta$ might be not present. Also, observe that, within $\ket{\Gamma}$, all possible solutions of the considered problems are included, if any exist.

Second, in Grover's approach, the {\sc Inversion about the Mean} circuit is prescribed. Hence, we now focus on the {\sc Phase Inversion} circuit. In the first iteration, it receives as input 
\begin{inparaenum}[\bf (i)]
\item the uniform superposition $\ket{\Gamma} = H^{\otimes \ell}\ket{0_{\ell}}$, 
\item $\alpha$ ancilla qubits set to $\ket{0}$, whose number depends on the type of problem we are addressing, and
\item a qubit set to $\ket{-}$. 
\end{inparaenum}
Namely, it receives as input the superposition $\ket{\Gamma}\ket{0_{\alpha}}\ket{-}$, where $\ket{\Gamma} = \frac{1}{\sqrt{2^{\ell}}}\sum_{\gamma \in {\mathbb B}^{\ell}} \ket{\gamma}$.
It outputs the superposition $\frac{1}{\sqrt{2^{\ell}}}\sum_{\gamma \in {\mathbb B}^\ell} (-1)^{f(\gamma)}\ket{\gamma} \ket{0_{\alpha}} \ket{-}$, where $f(\gamma)=1$ if and only if $\gamma$ represents a valid solution to the considered problem. 
In general, the {\sc Phase Inversion} circuit receives in input the superposition $\ket{\Gamma}\ket{0_{\alpha}}\ket{-}$, where $\ket{\Gamma} = \sum_{\gamma \in {\mathbb B}^{\ell}} c_\gamma \ket{\gamma}$.
It outputs the the superposition $\sum_{\gamma \in {\mathbb B}^{\ell}} (-1)^{f(\gamma)} c_\gamma\ket{\gamma} \ket{0_{\alpha}} \ket{-}$. We remark that the values of the complex coefficients $c_\gamma$ depend on the iteration.

For each problem we consider, we provide a specific {\sc Phase Inversion} circuit. All such circuits consist of three circuits (see  \cref{fig:framework}), the first is called {\sc Input Transducer} and is denoted by $U_I$, the second is called {\sc Solution Detector} and is denoted by $U_S$, and the third is the inverse $U_I^{-1}$ of the {\sc Input Transducer}. The purpose of the {\sc Input Transducer} circuits is to ``filter out'' the states of $\ket{\Gamma}$ that do not correspond to ``well-formed'' candidate solutions. The purpose of the {\sc Solution Detector} circuits is to invert the amplitude of the states of $\ket{\Gamma}$ that correspond to positive solutions, if any. The purpose of 
$U_I^{-1}$ is to restore the state of the ancilla qubits to $\ket{0}$ so that they may be employed in the subsequent iterations of the amplitude-amplification process.

The {\sc Input Transducer} circuits are described in \cref{se:input-initializers}.
The {\sc Solution Detector} circuits are described in \cref{se:solution-detectors}.
In that section, we also combine the {\sc Input Transducer} circuits and the {\sc Solution Detector} circuits to prove the following lemma; refer also to~\cref{tab:circuit-model-results}.

\begin{lemma}\label{lem:oracles}
The TLCM, TLKP, TLQP, TLS, OPCM, BT, and BS problems admit {\sc Phase Inversion} circuits whose circuit complexity, depth, and width are bounded as follows:

\begin{description}
\item[TLCM] Circuit complexity: $O(m^2)$. Depth: $O(n^2)$. Width $O(m^2)$. 
\item[TLKP] Circuit complexity: $O(m^2)$. Depth: $O(m \log^2 m)$. Width $O(m)$. 
\item[TLQP] Circuit complexity: $O(m^6)$. Depth: $O(m^4)$. Width $O(m^2)$. 
\item[TLS]  Circuit complexity: $O(m^2)$. Depth: $O(m)$. Width $O(m)$. 
\item[OPCM] Circuit complexity: $O(n^8)$. Depth: $O(n^6)$. Width $O(m^2)$. 
\item[BT]   Circuit complexity: $O(n^8)$. Depth: $O(n^6)$. Width $O(m)$. 
\item[BS]   Circuit complexity: $O(n^8)$. Depth: $O(n^6)$. Width $O(m)$. 
\end{description}

\end{lemma}

\cref{th:grover,lem:oracles} imply the following.

\begin{theorem}\label{th:main}
In the quantum circuit model of computation, the TLCM, TLKP, TLQP, TLS, OPCM, BT, and BS problems can be solved with the following sequential and parallel time bounds (where $M$ denotes the number of solutions to the problem):

\begin{description}
\item[TLCM] Sequential: $\sqrt{\frac{2^{n \log n}}{M}} \cdot O(m^2)$. Parallel: $\sqrt{\frac{2^{n \log n}}{M}} \cdot O(n^2)$. 
\item[TLKP] Sequential: $\sqrt{\frac{2^{n \log n}}{M}} \cdot O(m^2)$. Parallel: $\sqrt{\frac{2^{n \log n}}{M}} \cdot O(m \log^2 m)$. 
\item[TLQP] Sequential: $\sqrt{\frac{2^{n \log n}}{M}} \cdot O(m^6)$. Parallel: $\sqrt{\frac{2^{n \log n}}{M}} \cdot O(m^4)$. 
\item[TLS]  Sequential: $\sqrt{\frac{2^{n \log n + \sigma\log m}}{M}}\cdot O(m^2)$. Parallel: $\sqrt{\frac{2^{n \log n + \sigma\log m}}{M}}\cdot O(m)$. 
\item[OPCM] Sequential: $\sqrt{\frac{2^{n \log n}}{M}}\cdot O(n^8)$. Parallel: $\sqrt{\frac{2^{n \log n}}{M}}\cdot O(n^6)$. 
\item[BT]   Sequential: $\sqrt{\frac{2^{n \log n + m\log \tau}}{M}}\cdot O(n^8)$. Parallel: $\sqrt{\frac{2^{n \log n + m\log \tau}}{M}}\cdot O(n^6)$. 
\item[BS]   Sequential: $\sqrt{\frac{2^{n \log n + \sigma\log m}}{M}}\cdot O(n^8)$. Parallel: $\sqrt{\frac{2^{n \log n + \sigma\log m}}{M}}\cdot O(n^6)$. 
\end{description}
\end{theorem}

\section{Input Transducer Circuits}\label{se:input-initializers}

\begin{figure}[b!]
	\centering
	\includegraphics[page = 49, width = .6\textwidth]{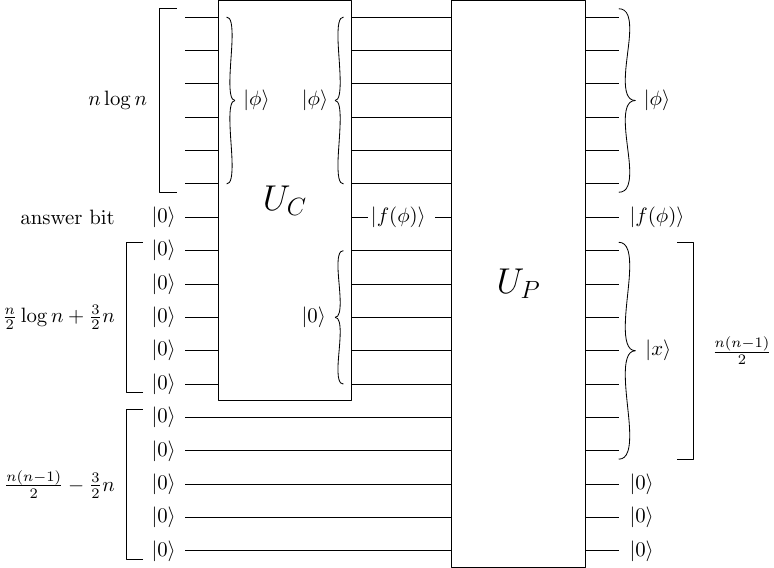}
	\caption{The {\sc Order Transducer} gate $U_{OT}$.}
	\label{fig:Order-Initializer}
\end{figure}

We use two different versions of circuit $U_I$, depending on the considered problem. Namely, for all problems but for the TLS and the BS problems, circuit $U_I$ consists of just one circuit $U_{OT}$, called \emph{\sc Order-Transducer} (refer to \cref{fig:Order-Initializer}). For problems TLS and BS, circuit $U_I$ executes, in parallel to $U_{OT}$, another circuit $U_{ST}$, called \emph{\sc Skewness-Transducer} (refer to \cref{fig:skewness-initializer}).

\subsection{Order Transducer}

Let $\phi$ be a binary string of length $n \log n$, which we interpret as a sequence of $n$ binary integers, each consisting of $\log n$ bits. 
Recall that, we denote by $\phi[i]$ the $i$-th binary integer contained in $\phi$. 
Also, let $\ket{\phi}$ be the basis state corresponding to $\phi$. This subsection is devoted to proving the following lemma.

\begin{lemma}\label{le:gate-order-initializer}
There exists a gate $U_{OT}$ that, when provided with the input superposition $\ket{\phi} \ket{0_{\alpha}}$, where $\alpha = \frac{n}{2}(n-1+\log n) +1$, produces the output superposition 
$$\ket{\phi} \ket{f(\phi)} \ket{x_{0,1}}\dots\ket{x_{0,n-1}}\dots\ket{x_{i,j}}\dots\ket{x_{n-2,n-1}} \ket{0_{\frac{n}{2}\log n}},$$ such that $\ket{x_{i,j}}=1$ if and only if $\phi[i] < \phi[j]$ and $f(\phi)=1$ if and only if $\phi$ represents an $n$-permutation. $U_{OT}$ has 
$O(n^2 \log n)$ circuit complexity, and $O(n\log n)$ depth and width.
\end{lemma}

\paragraph{\sc Proof of \cref{le:gate-order-initializer}.} The input of $U_{OT}$ is composed of $n \log n + \alpha$ qubits, the first $\ell=n\log n$ form a superposition $\ket{\phi}$ and the rest are set to $\ket{0}$.
First, $\ket{\phi}$ and $\frac{n}{2}(\log n +3) +1$ qubits set to $\ket{0}$ enter a gate $U_C$, 
called \emph{\sc Collision Detector}. The purpose of $U_C$ is to compute the superposition $\ket{\phi}\ket{f(\phi)}\ket{0_{ \frac{n}{2}(\log n +3)}}$, where $f(\phi)=1$ if and only if $\phi[i] \neq \phi[j]$ for each $i \neq j$. It has $O(n^2\log n)$ circuit complexity, and it has $O(n\log n)$ depth and width.
Second, $\ket{\phi}$ and $\frac{n}{2}\log n + \frac{n(n-1)}{2}$ qubits set to $\ket{0}$ enter a gate $U_P$ called \emph{\sc Precedence Constructor}. The purpose of $U_P$ is to compute a superposition $\ket{\phi}\ket{x}\ket{0_{\frac{n}{2} \log n}}$, where $\ket{x}=\ket{x_{0,1}}\dots\ket{x_{0,n-1}}\dots\ket{x_{i,j}}\dots\ket{x_{n-2,n-1}}$ and $x_{i,j}=1$ if and only if $\phi[i] < \phi[j]$. 
It has circuit complexity $O(n^2\log n)$, and it has $O(n\log n)$ depth and width.
Gate $U_{OT}$ has $O(n^2\log n)$ circuit complexity, and it has $O(n\log n)$ depth and width.

\paragraph{\sc Collision-detector.} Gate $U_C$ works as follows. Refer to \cref{fig:UC}.
\begin{figure}[tb!]
    \centering
    \includegraphics[page = 9, width = .8\textwidth]{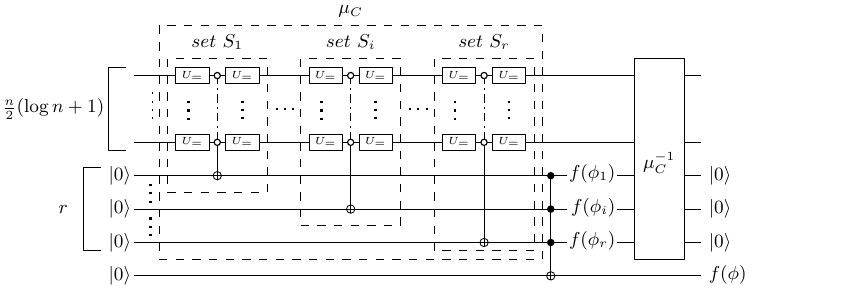}
    \caption{The gate $U_{C}$.}
    \label{fig:UC}
\end{figure}

It executes two gates $\mu_{C}$ and $\mu_{C}^{-1}$ and in between such gates it executes a Toffoli gate.
The input of $\mu_{C}$ is the superposition $\ket{\phi[0]} \dots \ket{\phi[n-1]}\ket{0_{\frac{n}{2}(\log n +1)}}\ket{0_n}$. The output of $\mu_{C}$ is the superposition $\ket{\phi[0]} \dots \ket{\phi[n-1]}\ket{0_{\frac{n}{2}(\log n +1)}}\ket{f(\phi_0)}\dots\ket{f(\phi_i)}\dots \ket{f(\phi_{n-1})}$.
In gate $\mu_{C}$, we compare the unordered pairs of numbers $\phi[i]$ and $\phi[j]$ in parallel as follows. Consider that if two numbers are compared,  none of the two can be compared with another number at the same time. Hence, we partition the pairs using \cref{le:partion-k} (with $k=2$ and $|X| = n$) into $r = O(n)$ cross-independent sets $S_1, \dots, S_{r}$ each containing at most $\frac{n}{2}$ pairs. %

For each pair $(i,j)$ of $S_1$ (refer to \cref{fig:UC}) we use a $U_{=}$ gate to compare $\phi[i]$ and $\phi[j]$. Recall that the gate $U_{=}$ outputs a superposition $\ket{\phi[i]}\ket{\phi[j]}\ket{0}\dots \ket{0}\ket{\phi[i]==\phi[j]}$. All the last output qubits of the $U_{=}$ gates for $S_1$ enter a Toffoli gate with $|S_i| +1$ inputs and outputs, which computes the qubit $\ket{f(\phi_0)}$ such that $\phi_0=1$ if and only if all of them first $|S_i|$ input qubits are equal to $\ket{0}$, i.e., all pairs correspond to different numbers.

After dealing with $S_1$, we deal with $S_2$ with the same technique and keep on dealing with the $S_i$ sets until $S_{r}$ is reached.

All the last output qubits of the $U_{=}$ gates for $S_i$ enter a Toffoli gate that outputs a qubit $\ket{f(\phi_i)}$ such that $f(\phi_1) = 1$ if and only if all of them are equal to $\ket{0}$, i.e., if does not exist pair $(j,k)$ of $S_i$ where $\phi[j]=\phi[k]$. In order to allow the reuse of the ancilla qubits, except for the qubit $\ket{f(\phi_i)}$, gate $U_C$ executes in parallel a gate $U^{-1}_{=}$ for each pair in $S_i$.

All the qubits $\ket{f(\phi_i)}$ and the qubit $\ket{\phi}$ enter a Toffoli gate with $r+1$ inputs and outputs. The first $r$ qubits are control qubits, the target qubit is $\ket{f(\phi)}$, which is initialized to $\ket{0}$. The target qubit is set to $\ket{\bigwedge^r_{i=1} \phi_i}$. 
In order to allow the reuse of the ancilla qubits, we apply to the entire circuit preceding the Toffoli gate its inverse gate. Recall that, by \cref{le:gate-equal}, gate $U_{=}$ has $O(\log n)$ circuit complexity, depth, and width.
Therefore, gate $U_{C}$ has $O(n^2 \log n)$ circuit complexity, and it has $O(n\log n)$ depth and width.

\paragraph{\sc Precedence Constructor.}

Gate $U_P$ works as follows. Refer to \cref{fig:precedence}.
\begin{figure}[tb]
    \centering
    \includegraphics[page = 10, width = .6\textwidth]{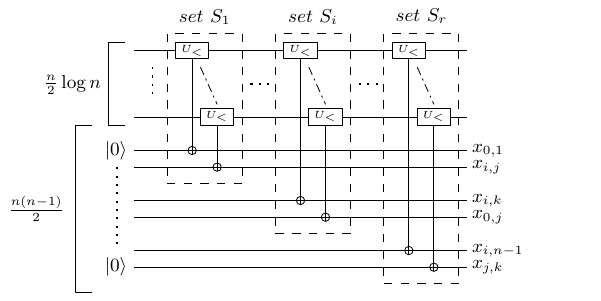}
    \caption{The gate $U_P$.}
    \label{fig:precedence}
\end{figure}

For each pair $i,j$ ($i,j \in [n]$ and $i \neq j$) $U_P$ exploits $U_{<}$ that outputs a qubit $\ket{x_{i,j}}$ such that $x_{i,j}=\phi[i]<\phi[j]$. %
Using several gates $U_{<}$, we compare the ordered pairs of numbers $\phi[i]$ and $\phi[j]$ in parallel as follows. As for gate $U_C$, if two numbers are compared,  none of the two can be compared with another number at the same time. Hence, we partition the pairs  
using \cref{le:partion-k} (with $k=2$ and $|X| = n$) into $r \in O(n)$ cross-independent sets $S_1, \dots, S_{r}$ each containing at most $\frac{n}{2}$ pairs.

For each pair $(i,j)$ of $S_1$ (refer to \cref{fig:precedence}) we use a $U_<$ gate to compare $\phi[i]$ and $\phi[j]$. Recall that the gate $U_<$ outputs a superposition $\ket{\phi[i]}\ket{\phi[j]}\ket{0}\dots \ket{0}\ket{\phi[i]<\phi[j]}$. In order to allow the reuse of the ancilla qubit different from $\ket{x_{i,j}}$ for each pair $(i,j) \in S_1$, we use a symmetric circuit that transforms the ancilla output qubits into a sequence of $\ket{0}$ qubits, that will be re-used in the following step.

After dealing with $S_1$, we deal with $S_2$ with the same technique and keep on dealing with the $S_i$ sets until $S_{r}$ is reached.

Recall that, by \cref{le:gate-less}, gate $U_<$ has $O(\log n)$ circuit complexity, depth, and width. Therefore, gate $U_P$ has $O(n^2 \log n)$ circuit complexity, and it has $O(n \log n)$ depth and width. 

\subsection{Skewness Transducer}

Let $\theta$ be a binary string of length $\sigma\log m$, which we interpret as a sequence of $\sigma$ binary integers, each consisting of $\log m$ bits. Recall that, we denote by $\theta[i]$ the $i$-th binary integer contained in~$\theta$. Also, let $\ket{\theta}$ be the basis state corresponding to $\theta$.
This subsection is devoted to proving the following lemma.

\begin{lemma}\label{le:gate-skewness-initializer} 
There exists a gate $U_{ST}$ that, when provided with the input superposition $\ket{\theta}\ket{0}\ket{0_{m}}$, 
produces the output superposition $\ket{\theta}\ket{f(\theta)}\ket{e_0}\dots\ket{e_i}\dots\ket{e_{m-1}}$, such that 
$f(\theta) = 1$ if and only if $\theta$ represents a subset of size $\sigma$ of the set $[m]$, and when $f(\theta)=1$ it holds that 
$e_i = 1$ if and only if there exists $j \in [\sigma]$ such that $\theta[j]$ coincides with (the binary representation of) the integer $i$.
Gate $U_{ST}$ has $O(\sigma m\log m)$ circuit complexity and depth, and $O(\sigma\log m)$ width.
\end{lemma}

\begin{figure}[tb!]
    \centering
    \includegraphics[page = 50, scale = 0.8]{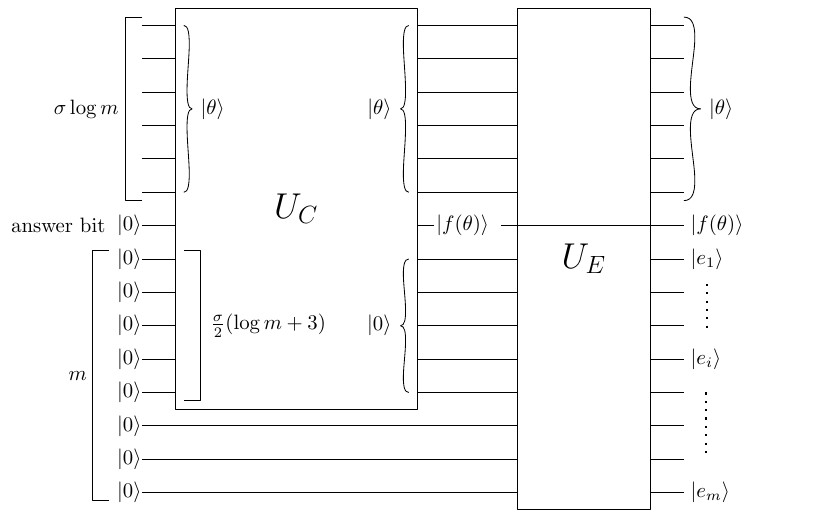}
    \caption{The {\sc Skewness Transducer} gate $U_{ST}$.}
    \label{fig:skewness-initializer}
\end{figure}

\paragraph{\sc Proof of \cref{le:gate-skewness-initializer}.} 
The input of $U_{ST}$ is composed of $\sigma \log m + 1 + m$ qubits, the first $\ell = \sigma\log m$ qubits form a superposition $\theta$ and the rest are set to $\ket{0}$. First, $\ket{\theta}$, $\ket{0}$, and $h=\frac{\sigma}{2}(\log m + 3)$ qubits set to $\ket{0}$ enter an instance of the {\sc Collision Detector} gate $U_C$ used in the proof of \cref{le:gate-order-initializer}, where the qubits of the superposition $\ket{\theta}\ket{0}\ket{0_h}$ play the role of the qubits of the superposition $\ket{\phi}\ket{0}\ket{0_{\frac{n}{2}(\log n +3)}}$. 
The purpose of this instance of $U_C$ is to compute the superposition $\ket{\theta}\ket{f(\theta)}\ket{0_{h}}$, where $f(\theta) = 1$ if and only if $\theta[a]\neq \theta[b]$ for each $0\leq a<b \leq \sigma-1$.
It has $O(\sigma^2 \log m)$ circuit complexity and it has $O(\sigma\log(m))$ depth and width.
Second, the superpositions $\ket{\theta}$ and $\ket{0_{m}}$ enter a gate $U_E$, called \emph{\sc Edge Constructor}; refer to \cref{fig:U_e}. The purpose of $U_E$ is to compute the superposition $\ket{\theta}\ket{e}$, where $\ket{e}=\ket{e_0}\dots\ket{e_i}\dots\ket{e_{m-1}}$ and, when $f(\theta)=1$, it holds that $e_i=1$ if and only if there exists a $j\in [\sigma]$ such that $\theta[j]$ coincides with (the binary representation of) the integer $i$. 
It has $O(\sigma m\log m)$ circuit complexity, $O(\sigma m)$ depth, and $O(\log m)$ width.
Recall that, $\sigma <m$. Thus, we have that
gate $U_{ST}$ has  $O(\sigma m\log m)$ circuit complexity, $O(\sigma m)$ depth, and $O(\sigma \log m)$ width.

\paragraph{\sc Edge Constructor.}
The gate $U_E$ exploits instances of the auxiliary gate $U_{e_i}$, defined for each edge $e_i \in E$ as follows. Refer to \cref{fig:U_e_i}.
When provided with the input superposition $\ket{\theta}\ket{0}$, the gate $U_{e_i}$ produces the output superposition $\ket{\theta}\ket{e_i}$, where -- provided that $\theta$ represents a subset of size $\sigma$ of the set $[m]$, i.e., $f(\theta)=1$ -- it holds that
$e_i=1$ if and only if there exists $j \in [\sigma]$ such that $\theta[j]$ coincides with the integer $i$. 
The gate $U_{e_i}$ contains $\sigma$ Toffoli gates, each with $\log(m) + 1$ inputs and outputs. All such Toffoli gates share the same target qubit.  
The control qubits of the first Toffoli gate $T_0$ are $\ket{\theta[0][0]}\dots\ket{\theta[0][\log(m)-1]}$ and its target qubit is initialized to $\ket{0}$. It turns the target qubit into $\ket{1}$ if and only if $\theta[0][0]\dots\theta[0][\log(m)-1]$ coincides with the integer $i$. For $j=1,\dots,\sigma-1$, gate $U_{e_i}$ contains a Toffoli gate $T_j$ that takes in input the superposition
$\ket{\theta[j][0]}\dots\ket{\theta[j][\log(m)-1]}$ and the target qubit that has been output by $T_{j-1}$. It flips the value of the target qubit if and only if $\theta[j][0]\dots\theta[j][\log(m)-1]$ coincides with  (the binary representation of) the integer $i$. Thus, if at most one of the $j$ integers composing~$\theta$ coincides with the integer $i$, then  $e_i = 1$ if and only if there exists $j \in [\sigma]$ such that $\theta[j]$ coincides with (the binary representation of) the integer $i$.
Gate $U_{e_i}$ has $O(\sigma\log m)$ circuit complexity, $O(\sigma)$ depth, and $O(\log m)$ width.

The gate $U_E$ computes $\ket{e}=\ket{e_0}\dots\ket{e_i}\dots\ket{e_{m-1}}$ as follows. For $i \in [m]$, it computes $\ket{e_i}$ by applying the gate $U_{e_i}$ to the superposition $\ket{\theta}\ket{0}$, where the last qubit is the $i$-th qubit of the $m$ qubits compositing the superposition $\ket{0_m}$, which is provided in input to the gate $U_E$. 
Gate $U_E$ has circuit complexity $O(\sigma m\log m)$, depth $O(\sigma m)$, and width $O(\log m)$.

\begin{figure}[tb!]
        \begin{subfigure}{0.5\textwidth}        
        \centering
        \includegraphics[page = 39, width = 0.95\textwidth]{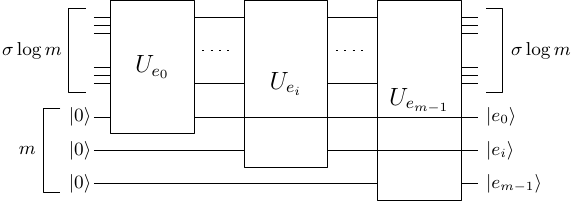}
        \caption{Gate $U_E$}
        \label{fig:U_e}
        \end{subfigure}
        ~~~~~
        \begin{subfigure}{0.5\textwidth}        
        \centering
        \includegraphics[page = 59, width = 0.95\textwidth]{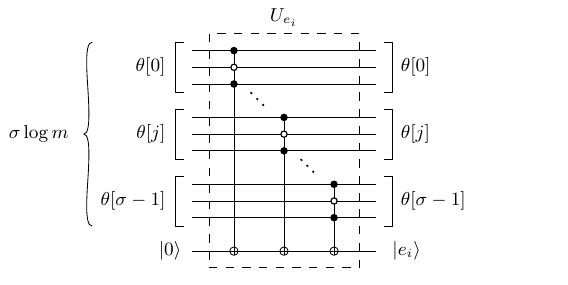}
        \caption{Gate $U_{e_i}$}
        \label{fig:U_e_i}
        \end{subfigure}
        \caption{Illustration for the construction of gate $U_E$.}
        \label{fig:U_ei_U_e}
\end{figure}

\section{Solution Detector Circuits}\label{se:solution-detectors}

In this section, we present the details of the {\sc Solution Detector} circuits $U_S$ for the problems we consider. %

Recall that, the {\sc Order Trasducer} circuit produces the output superposition 
$$\ket{\phi} \ket{f(\phi)} \ket{x_{0,1}}\dots\ket{x_{0,n-1}}\dots\ket{x_{i,j}}\dots\ket{x_{n-2,n-1}} \ket{0_{\frac{n}{2}\log n}},$$ such that $x_{i,j}=1$ if and only if $\phi[i] < \phi[j]$, and $f(\phi)=1$ if and only if $\phi$ represents an $n$-permutation. In the following, for simplicity, we denote the superposition $\ket{x_{0,1}}\dots\ket{x_{0,n-1}}\dots\ket{x_{i,j}}\dots\ket{x_{n-2,n-1}}$ as $\ket{x}$. 
We interpret the values $x_{0,1},\dots, x_{0,n-1},\dots, x_{i,j}, \dots, x_{n-2,n-1}$ as the entries above the main diagonal of a square binary matrix $X$, where $X[i][j] = x_{i,j}$, and whose entries along and below the main diagonal are undefined.
We use such entries to represent the precedence between vertices in a graph drawing. Let $x$ be the string obtained by concatenating $x_{0,1},\dots, x_{0,n-1},\dots, x_{i,j}, \dots, x_{n-2,n-1}$. 
We will use $x$ both for book layouts and for $2$-level drawings as follows.

Consider book layouts of a graph $G = (V,E)$.
We denote by $\Pi(x)$ the vertex order along the spine of a book layout of $G$ defined as follows. We have that, if $x_{i,j}=1$, then vertex $v_i$ precedes vertex $v_j$ in $\Pi(x)$. Conversely, if $x_{i,j}=0$,  then vertex $v_j$ precedes vertex $v_i$ in $\Pi(x)$.
Consider now a $2$-level drawing of a graph $G = (U,V,E)$. We assume that the vertices in $U$ are labeled as $u_0,\dots,u_{|U|-1}$, and the vertices of $V$ are labeled as $v_{|U|},\dots,v_{|U|+|V|-1}$.
We denote by $D(x)$ the $2$-level drawing of $G$ defined as follows. Let $w_i$ and $w_j$ be two vertices of $G$. 
Suppose that $w_i, w_j \in U$. If $x_{i,j}=1$, then  $w_i \prec w_j$ along the horizontal line $\ell_u$, otherwise, $w_j \prec w_i$ along $\ell_u$.
Suppose that $w_i, w_j \in V$. If $x_{i,j}=1$, then  $w_i \prec w_j$ along the horizontal line $\ell_v$, otherwise, $w_j \prec w_i$ along $\ell_v$.
Suppose now that $w_i \in U$ and $w_j \in V$. Then, we assume that  $x_{i,j}=0$, which we interpret as the absence of a precedence relation between such vertices.

We remark that, if $\phi$ is not an $n$-permutation, then $\Pi(x)$ and $D(x)$ do not correspond to actual spine orders and $2$-level drawings, respectively. In this case, we say that they are \emph{degenerate}.

We will exploit $\ket{x}$ to compute a superposition $\ket{\chi_{0,1}}\dots\ket{\chi_{0,m-1}}\dots\ket{\chi_{i,j}}\dots\ket{\chi_{m-2,m-1}}$, which we will denote for simplicity by $\ket{\chi}$.
We interpret the values $\chi_{0,1},\dots, \chi_{0,m-1},\dots, \chi_{i,j}, \dots, \chi_{m-2,m-1}$ as the entries of a square binary matrix $C$, where $C[i][j] = \chi_{i,j}$ and whose entries along and below the main diagonal are undefined. We use such entries to represent the existence of crossings between pairs of edges in a graph drawing. Namely, $\chi_{i,j}=1$ if $e_i$ and $e_j$ belong to $E$ and cross, and $\chi_{i,j} = 0$ if either $e_i$ and $e_j$ belong to $E$ and do not cross or at least one of $e_i$ and $e_j$ does not belong to $E$. Let $\chi$ be the string obtained by concatenating $\chi_{0,1},\dots, \chi_{0,m-1},\dots, \chi_{i,j}, \dots, \chi_{m-2,m-1}$. 
The values of $\chi$ are completely determined by $x$ and by whether the considered layout is a book layout or a $2$-level drawing. 
For every $0 \leq a < b \leq m-1$, consider the value $\chi_{a,b}$, where $e_a = (v_i,v_k)$ and $e_b = (v_j,v_\ell)$.
In a book layout of $G$ in which the vertex order is $\Pi(x)$, we have that $\chi_{a,b} = 1$ if $e_a$ and $e_b$ belong to $E$ and cross (refer to the conditions in \cref{fig:book-crossing}), and $\chi_{a,b} = 0$ if either $e_a$ and $e_b$ belong to $E$ and do not cross or at least one of $e_a$ and $e_b$ does not belong to $E$.
If $D(x)$ corresponds to a $2$-level drawing of $G$, then we have that $\chi_{a,b} = 1$ if $e_a$ and $e_b$ belong to $E$ and cross (i.e., $x_{i,j} \neq x_{k,\ell}$), and $\chi_{a,b} = 0$ if either $e_a$ and $e_b$ belong to $E$ and do not cross (i.e., $x_{i,j} = x_{k,\ell}$) or at least one of $e_a$ and $e_b$ does not belong to $E$.

Recall that, the {\sc Skewness Trasducer} circuit outputs the superposition 
$\ket{\theta} \ket{f(\theta)} \ket{e_0}\dots\ket{e_i}\dots\ket{e_{m-1}}$ such that, when $f(\theta)=1$, it holds that $e_i=1$ if and only if there exists a $j \in [\sigma]$ such that $\theta[j]$ coincides with (the binary representation of) the integer $i$. In the following, for simplicity, we denote the superposition $\ket{e_0}\dots\ket{e_i}\dots\ket{e_{m-1}}$ as $\ket{e}$. 
Observe that, during the computation for problems TLS and BS, we manage the superposition $\ket{\Theta} = \sum_{\theta \in {\mathbb B}^{\sigma \log m}} c_\theta \ket{\theta}$, which includes all possible states corresponding to a selection of $\sigma$ edges of $G$.
Specifically, consider any basis state $\theta$ that appears in $\ket{\Theta}$, which represents a (multi)subset $N(\theta)$ of size $\sigma$ of the set $[m]$. Recall that, the integers contained in $\theta$ are the labels of the edges of $G$. 
We denote by $K(\theta)$ the subset of the edges of $G$ whose indices appear in $N(\theta)$.
Observe that, if $N(\theta)$ does not contain repeated entries, then $K(\theta)$ is a subset of $\sigma$ edges of $G$ (with no repeated edges); this occurs exactly when $f(\theta) = 1$. If $N(\theta)$ contains repeated entries, then we say that it is \emph{degenerate}. %

\subsection{Problem TLCM} We call TLCM the {\sc Solution Detector} circuit for problem TLCM.
Recall that, for the TLCM problem, we denote by $\rho$ the maximum number of crossings allowed in the sought $2$-level drawing of $G$.

\begin{lemma}\label{le:gate-TLCM}
There exists a gate TLCM that, when provided with the input superposition 
$\ket{f(\phi)} \ket{x} \ket{0_{h}} \ket{-}$, where $h=5\frac{m(m-1)}{2} - \log m -\log(m-1)-2$, produces the output superposition 
$(-1)^{g(x)f(\phi)}\ket{f(\phi)}\ket{x} \ket{0_{h}}\ket{-}$, where $g(x) = 1$ if $D(x)$ is not degenerate and the $2$-level drawing $D(x)$ of $G$ has at most $\rho$ crossings.  
TLCM has $O(m^2)$ circuit complexity, $O(n^2)$ depth, and $O(m^2)$ width. 
\end{lemma}

\paragraph{\bf Proof of \cref{le:gate-TLCM}.} 
\begin{figure}[tb!]
    \centering
    \includegraphics[page = 29, width = .9\textwidth]{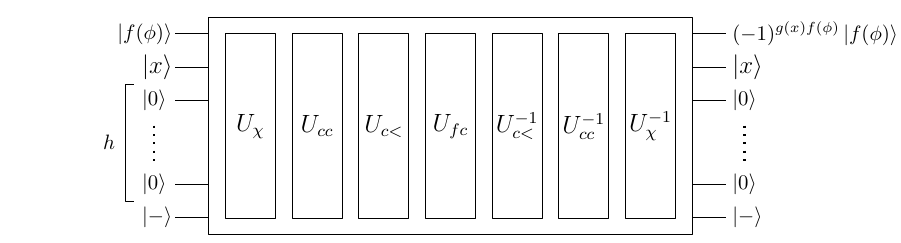}
    \caption{TLCM Oracle Pipeline.}
    \label{fig:TLCM oracle}
\end{figure}
Gate TLCM uses four gates: {\sc tl-Cross Finder} $U_\chi$, {\sc Cross Counter} $U_{cc}$, {\sc Cross Comparator} $U_{c<}$, and {\sc Final Check} $U_{fc}$, followed by the inverse gates $U^{-1}_{c<}$, $U^{-1}_{cc}$, and $U^{-1}_{\chi}$. Refer to \cref{fig:TLCM oracle}.

\paragraph{\sc tl-cross finder.}

\begin{figure}[htb!]
        \centering
        \begin{subfigure}{0.15\textwidth}        
        \centering
        \includegraphics[page = 15]{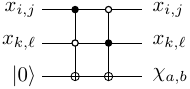}
        \caption{Gate $U_{cr}$}
        \label{fig:U_cr}
        \end{subfigure}
        ~~~~~~~~~~
        \begin{subfigure}{0.45\textwidth}        
        \centering
        \includegraphics[page = 17]{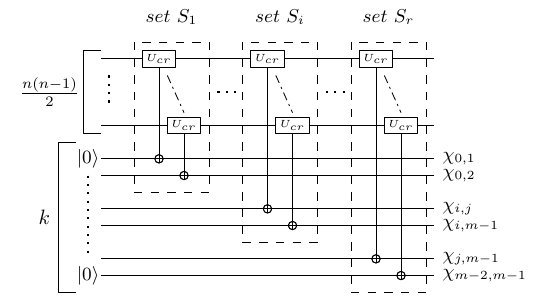}
        \caption{Gate $U_{\chi}$}
        \label{fig:U_chi}
        \end{subfigure}
        \caption{The gate $U_{cr}$ (a) and gate $U_{\chi}$ (b). In (b), it holds $k=\frac{m(m-1)}{2}$.}
        \label{fig:U_cr-U_chi}
\end{figure}

The purpose of $U_{\chi}$ is to compute the crossings in $D(x)$ (under the assumption that $D(x)$ is not degenerate), determined by the vertex order corresponding to $x$; refer to~\cref{fig:U_chi}. When provided with the input superposition $\ket{x}\ket{0_k}$, where $k=\frac{m(m-1)}{2}$, the gate $U_{\chi}$ produces the output superposition $\ket{x}\ket{\chi}$. 

The gate $U_{\chi}$ exploits the auxiliary gate $U_{cr}$, whose purpose is to check if two edges cross; refer to~\cref{fig:U_cr}. When provided with the input superposition $\ket{x_{i,j}}\ket{x_{k,\ell}}\ket{0}$, the gate produces the output superposition $\ket{x_{i,j}}\ket{x_{k,\ell}}\ket{\chi_{a,b}}$, where $e_a = (u_i,v_k)$,  $e_b = (u_j,v_\ell)$, and $\chi_{a,b} = x_{i,j} \oplus x_{k,\ell}$ (which is $1$ if and only if $e_a$ and $e_b$ cross in $D(x)$). %
It is implemented using two Toffoli gates with three inputs and outputs. The first one is activated when the qubit $\ket{x_{i,j}}$ is equal to $\ket{1}$ and the qubit $\ket{x_{k,\ell}}$ is equal to $\ket{0}$. The second one is activated when the qubit $\ket{x_{i,j}}$ is equal to $\ket{0}$ and the qubit $\ket{x_{k,\ell}}$ is equal to $\ket{1}$. 
 $U_{cr}$ has $O(1)$ circuit complexity, depth, and width.

The gate $U_{\chi}$ works as follows. Consider that if two variables $x_{i,j}$ and $x_{k,\ell}$ are compared to determine whether the edges $(u_i,v_k)$ and $(u_j,v_\ell)$ cross, none of these variables can be compared with another variable at the same time. Therefore, we partition the pairs of such variables using \cref{le:partion-k} (with $k=2$ and $|X| = \frac{n(n-1)}{2}$) into $r \in O(n^2)$ cross-independent sets $S_1, \dots, S_{r}$ each containing at most $\frac{n(n-1)}{4}$ pairs.  
For $i=1,\dots,r$, the gate $U_{\chi}$ executes in parallel a $U_{cr}$ gate, for each pair $(x_{i,j},x_{k,\ell})$ in $S_i$ (refer to \cref{fig:U_chi}), in order to output the qubit $\ket{\chi_{a,b}}$.
$U_{\chi}$ has $O(n^4)$ circuit complexity, $O(n^2)$ depth, and $O(n^2)$ width.

\paragraph{\sc cross counter.}

The purpose of gate $U_{cc}$ is to count the total number of crossings in the drawing $D(x)$. When provided with the input superposition $\ket{\chi}\ket{0_h}\ket{0_k}$, where $h=\log m + \log(m-1)$ and $k=2m(m-1) -2(\log m + \log(m-1)) -2$, the gate $U_{cc}$ produces the output superposition $\ket{\chi}\ket{\sigma(x)}\ket{0_h}$, where $\sigma(x) = \sum_{e_i,e_j \in E} \chi_{i,j}$ is a binary integer of length $\log m + \log(m-1)$ representing the total number of crossings.
The gate $U_{cc}$ is an instance of the gate $U_{1s}$ (refer to \cref{fig:U_s1}), where the qubits of the superposition $\ket{\chi} = \ket{\chi_{0,1}}\dots\ket{\chi_{0,m-1}}\dots\ket{\chi_{i,j}}\dots\ket{\chi_{m-2,m-1}}$ play the role of the qubits of the superposition $\ket{b_0 \dots b_{t-1}}$, where $t=m$, which forms part of the input of $U_{1s}$. 
Observe that the number of crossings in $D(x)$ is at most $\frac{m(m-1)}{2}$. Therefore, the number $\sigma(x)$ of crossings in $D(x)$ can be represented by a binary string of length $h \leq \log m + \log(m-1)$. By \cref{le:gate-counter}, replacing $t =\frac{m(m-1)}{2}$, we get that the parameter $k=2m(m-1) -2(\log m + \log(m-1)) -2$.
By \cref{le:gate-counter}, gate $U_{cc}$ has $O(m^2)$ circuit complexity, $O(\log^2 m)$ depth, and $O(m^2)$ width. %

\paragraph{\sc cross comparator.}

The purpose of gate $U_{c<}$ is to verify if the total number of crossings $\sigma(x)$ in $D(x)$ computed by gate $U_{cc}$ is less than the allowed number of crossings $\rho$. When provided with the input superposition $\ket{\sigma(x)}\ket{\rho}\ket{0_h}\ket{0}$, where $h=\log m + \log(m-1)$, the gate $U_{c<}$ produces the output superposition $\ket{\sigma(x)}\ket{\rho}\ket{0_h}\ket{g(x)}$, where $g(x)=1$ if $D(x)$ is not degenerate and $\sigma(x) < \rho$. The gate $U_{c<}$ is an instance of the gate $U_{<}$ (refer to \cref{fig:less}), where the $h=\log m + \log(m-1)$ qubits of the superposition $\ket{\sigma(x)}$ play the role of the qubits of the superposition $\ket{\phi[i]}$ and where $h$ qubits initialized to the binary representation of $\rho$ play the role of $\ket{\phi[j]}$. 
By \cref{le:gate-less}, gate $U_{c<}$ has $O(\log m)$ circuit complexity, depth, and width.

\paragraph{\sc final check.}

\begin{figure}[tb!]
    \centering
    \includegraphics[page = 21, width = .4\textwidth]{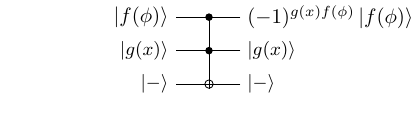}
    \caption{Gate $U_{fc}$.}
    \label{fig:final-check}
\end{figure}

The purpose of gate $U_{fc}$ is to check whether the current solution is admissible, i.e., whether the $2$-level drawing $D(x)$ of $G$ is not degenerate and it has at most $\rho$ crossings. Refer to \cref{fig:final-check}. When provided with the input superposition $\ket{f(\phi)}\ket{g(x)}\ket{-}$, the gate $U_{fc}$ produces the outputs superposition $(-1)^{g(x)f(\phi)}\ket{f(\phi)}\ket{g(x)}\ket{-}$. $U_{fc}$ exploits a Toffoli gate with three inputs and outputs. The control qubits are $\ket{f(\phi)}$ and $\ket{g(x)}$, and the target qubit is $\ket{-}$.
When $f(\phi)=g(x)=1$, the target qubit is transformed into the qubit $-\ket{-}$. Otherwise it leaves unchanged. %
Gate $U_{fc}$ has $O(1)$ circuit complexity, depth, and width. 

\paragraph{The inverse circuits.} The purpose of circuits $U^{-1}_{c<}$, $U^{-1}_{cc}$, and $U^{-1}_{\chi}$ is to restore the $h$ ancilla qubit to $\ket{0}$ so that they can be used in the subsequent steps of Grover's approach.

\paragraph{Correctness and complexity.}
For the correctness of \cref{le:gate-TLCM}, observe that the gates $U_\chi$, $U_{cc}$, $U_{c<}$, and $U_{fc}$ verify all the necessary conditions for which $D(x)$ has at most $\rho$ crossings, under the assumption that $D(x)$ is not degenerate. Therefore, the sign of the output superposition of gate TLCM, which is determined by the expression $(-1)^{g(x)f(\phi)}$, is positive when either $D(x)$ is degenerate or $D(x)$ is not degenerate and the number of crossings $\sigma(x)$ in $D(x)$ is larger than $\rho$, and it is negative when $D(x)$ is not degenerate and the number of crossings $\sigma(x)$ in $D(x)$ is smaller than $\rho$. 
The bound on the circuit complexity descends from the circuit complexity of the gate $U_{cc}$, the bound on the depth descends from the depth of the gate $U_{\chi}$, and the bound on the width descends from the width of $U_{cc}$.\qed

\subsection{Problem TLKP} We call TLKP the {\sc Solution Detector} circuit for problem  TLKP. Recall that, for TLKP problem, we denote by $k$ the maximum number of crossings allowed for each edge in the sought $2$-level drawing of $G$.

\begin{lemma}\label{le:gate-TLKP}
There exists a gate TLKP that, when provided with the input superposition 
$\ket{f(\phi)} \ket{x} \ket{0_{h}} \ket{-}$, where $h=\frac{m(m-1)}{2} + 4m - 2\log m -4 + m(\log m +1)$, outputs the superposition 
$(-1)^{g(x)f(\phi)}\ket{f(\phi)} \ket{x} \ket{0_{h}}\ket{-}$, where $g(x)=1$ if $D(x)$ is not degenerate and each edge of the $2$-level drawing $D(x)$ of $G$ has at most $k$ crossings. TLKP has
$O(m^2)$ circuit complexity, $O(m\log^2 m)$ depth and $O(m)$ width. 
\end{lemma}

\paragraph{\bf Proof of \cref{le:gate-TLKP}.} 
Gate TLKP uses four gates: {\sc tl-cross finder} $U_{\chi}$, {\sc edge cross counter} $U_{ecc}$, {\sc edge cross comparator} $U_{<k}$, and {\sc final check} $U_{fc}$, followed by the inverse gates $U_{\chi}^{-1}$, $U_{ecc}^{-1}$, and $U_{<k}^{-1}$. Refer to \cref{fig:TLKP-oracle}.

\begin{figure}[tb!]
    \centering
    \includegraphics[page = 30, width = .9\textwidth]{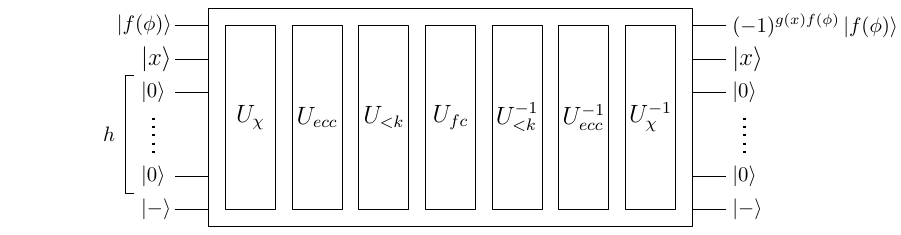}
    \caption{TLKP Oracle Pipeline.}
    \label{fig:TLKP-oracle}
\end{figure}

\paragraph{\sc tl-cross finder.} For the definition of gate $U_\chi$, refer to the proof of \cref{le:gate-TLCM}. Recall that the purpose of $U_{\chi}$ is to compute the crossings in $D(x)$ (under the assumption that $D(x)$ is not degenerate), determined by the vertex order corresponding to $x$. Also recall that, when provided with the input superposition $\ket{x}\ket{0_k}$, the gate $U_{\chi}$ produces the output superposition $\ket{x}\ket{\chi}$. 

\paragraph{\sc edge cross counter.}

The purpose of gate $U_{ecc}$ is to count, for each edge $e_i\in E$, the total number of crossings of each edge $e_i$ in the drawing $D(x)$. When provided with the input superposition $\ket{\chi}\ket{0_h}\ket{0_l}$, where $h=4m-2\log m -4$ and $l=m(1 + \log m)$, the gate $U_{ecc}$ produces the output superposition $\ket{\chi}\ket{0_h}\ket{\sigma_{e_0}(x)}\dots\ket{\sigma_{e_i}(x)}\dots\ket{\sigma_{e_{m-1}}(x)}$, where $\sigma_{e_i}(x)=\sum_{a < i} \chi_{a,i} + \sum_{b > i} \chi_{i,b}$ is a binary integer of length $1+\log m$ representing the total number of crossings of the edge $e_i$ in $D(x)$.
The gate $U_{ecc}$ exploits the auxiliary gate $U_{cc(e_i)}$, whose purpose, for each edge $e_i$, is to compute the binary integer $\sigma_{e_i}(x)$. When provided with the input superposition $\ket{\chi_{0,i}}\dots\ket{\chi_{i,i+1}}\dots\ket{\chi_{i,m-1}}\ket{0_h}\ket{0_{1+\log m}}$, the gate $U_{cc(e_i)}$ produces the output superposition $\ket{\chi_{0,i}}\dots\ket{\chi_{i,i+1}}\dots\ket{\chi_{i,m-1}}\ket{0_h}\ket{\sigma_{e_i}(x)}$. The gate $U_{cc(e_i)}$ is an instance of the gate $U_{1s}$ (refer to \cref{fig:U_s1}), where the qubits of the superposition $\ket{\chi_{0,i}}\dots\ket{\chi_{i,i+1}}\dots\ket{\chi_{i,m-1}}$ play the role of the qubits of the superposition $\ket{b_0\dots b_{t-1}}$, with $t=m$, which forms part of the input of $U_{1s}$. Observe that, for each edge $e_i$, the number of crossings of $e_i$ in $D(x)$ is at most $m-1$. Therefore, $\sigma_{e_i}(x)$ can be represented by a binary string of length $1 +\log m$. By \cref{le:gate-counter}, gate $U_{cc(e_i)}$ has circuit complexity $O(m)$, depth complexity $O(\log^2m)$, and width complexity $O(m)$.
Gate $U_{ecc}$ works as follows; refer to~\cref{fig:U_ecc}. Gate $U_{ecc}$ executes in sequence gates $U_{cc(e_0)}$, $U_{cc(e_1)}, \dots, U_{cc(e_{m-1})}$. Gate $U_{ecc}$ has circuit complexity $O(m^2)$, depth $O(m\log^2 m)$, and width $O(m)$.

\begin{figure}[tb!]
    \centering
    \includegraphics[page = 46, width = .8\textwidth]{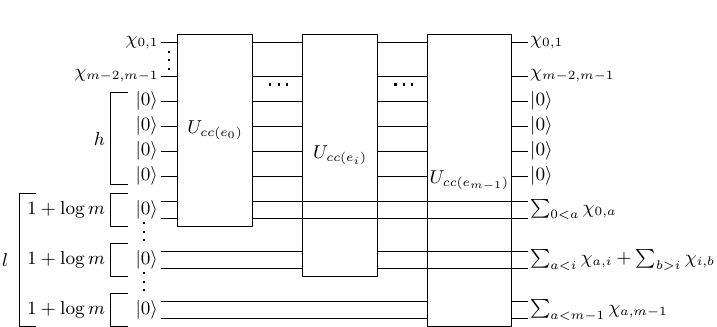}
    \caption{The gate $U_{ecc}$.}
    \label{fig:U_ecc}
\end{figure}

\paragraph{\sc edge cross comparator.}

The purpose of $U_{<k}$ is to verify, for each edge $e_i \in E$, if the total number of crossings $\sigma_{e_i}(x)$ of $e_i$ in $D(x)$ computed by $U_{ecc}$ is less than the allowed number of crossings for each edge $k$; refer to \cref{fig:U_edge_cross_comparator}. When provided with the input superposition $\ket{\sigma_E(x)}\ket{k}\ket{0_h}\ket{0_{m}}\ket{0}$, where $\sigma_E(x) = \sigma_{e_0}(x)\dots \sigma_{e_i}(x) \dots \sigma_{e_{m-1}}(x)$ and $h=\log m + 1$, the gate $U_{<k}$ produces the output superposition $\ket{\sigma_E(x)}\ket{k}\ket{0_h}\ket{0_{m}}\ket{g(x)}$, where $g(x)=1$ if $D(x)$ is not degenerate and $\sigma_{e_i}(x)<k$ for each $e_i\in E$. 
The gate $U_{<k}$ exploits $m$ instances of gate $U_<$ (refer to \cref{fig:less} and to \cref{le:gate-less}), where the qubits of the superposition $\ket{\sigma_{e_i}(x)}$ play the role of the qubits of the superposition $\ket{\phi[i]}$ and the qubits initialized to the binary representation of $k$ play the role of $\ket{\phi[j]}$. Recall that, by \cref{le:gate-less}, gate $U_<$ has $O(\log m)$ circuit complexity, depth, and width. Each of the $m$ gates $U_<$ provides an answer qubit ($\ket{\kappa_i}$), which is equal to $1$ if and only if $\sigma_{e_i}(x) < k$ for the considered edge $e_i$.
At the end of the $m$-th computation a Toffoli gate with $m+1$ inputs and outputs is applied to check if, for each $e_i\in E$, $\sigma_{e_i}(x)<k$. %
Gate $U_{<k}$ has $O(m\log m)$ circuit complexity, $O(m\log m)$ depth, and $O(m)$ width.

\begin{figure}[tb!]
    \centering
    \includegraphics[page = 34, width = 0.8\textwidth]{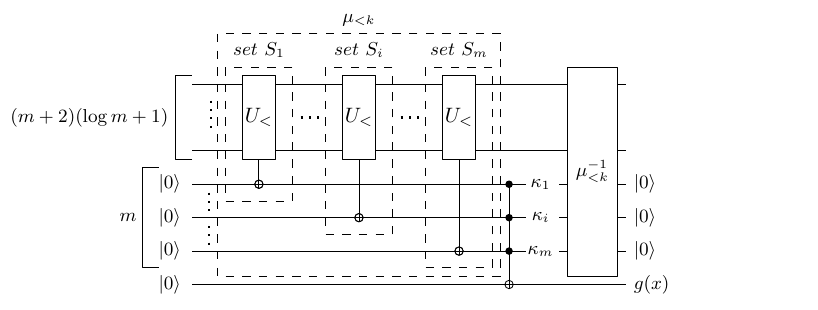}
    \caption{$U_{<k}$.}
    \label{fig:U_edge_cross_comparator}
\end{figure}

\paragraph{\sc final check.}

The purpose of gate $U_{fc}$ is to check wheter the current solution is admissibile, i.e., whether the $2$-level drawing $D(x)$ of $G$ is not degenerate and each edge has at most $k$ crossings. Refer to \cref{fig:final-check}. When provided with the input superposition $\ket{f(\phi)}\ket{g(x)}\ket{-}$, the gate $U_{fc}$ produces the outputs superposition $(-1)^{f(\phi)g(x)}\ket{-}$. $U_{fc}$ exploits a Toffoli gate with three inputs and outputs. The control qubits are $\ket{f(\phi)}$ and $\ket{g(x)}$, and the target qubit is $\ket{-}$. When at least one of $f(\phi)$ and $g(x)$ are equal to $0$, the target qubit leaves unchanged. On the other hand, when $f(\phi)=g(x)=1$, the target qubits is transformed into the qubit $-\ket{-}$. Gate $U_{fc}$ has $O(1)$ circuit complexity, depth and width.

\paragraph{The inverse circuits.} The purpose of circuits $U^{-1}_{{<k}}$, $U^{-1}_{cc(E)}$, and $U^{-1}_{\chi}$ is to restore the $h$ ancilla qubit to $\ket{0}$ so that they can be used in the subsequent steps of Grover's approach.

\paragraph{Correctness and complexity.}
For the correctness of \cref{le:gate-TLKP}, observe that the gates $U_\chi$, $U_{cc(E)}$, $U_{<k}$, and $U_{fc}$ verify all the necessary conditions for which $D(x)$, for each edge of $G$, has at most $k$ crossings, under the assumption that $D(x)$ is not degenerate. Therefore, the sign of the output superposition of gate TLKP, which is determined by the expression $(-1)^{g(x)f(\phi)}$, is positive when either $D(x)$ is degenerate or $D(x)$ is not degenerate and the number of crossings $\sigma_{e_i}(x)$ in $D(x)$ is larger than $k$, for some edge $e_i\in E$, and it is negative when $D(x)$ is not degenerate and the number of crossings $\sigma_{e_i}(x)$ in $D(x)$ is smaller than $k$, for each edge $e_i\in E$.
The bounds on the circuit complexity, depth, and width descend from the circuit complexity, depth, and width of the gate $U_{cc(E)}$.\qed

\subsection{Problem TLQP} We call TLQP the {\sc Solution Detector} circuit for problem TLQP.

\begin{lemma}\label{le:gate-TLQP}
There exists a gate TLQP that, when provided with the input superposition 
$\ket{f(\phi)} \ket{x} \ket{0_{h}} \ket{-}$, where 
$h \in O(m^4)$,
produces the output superposition 
$(-1)^{f(\phi)g(x)}\ket{f(\phi)} \ket{x} \ket{0_{h}}\ket{-}$, and $g(x) = 1$ if $D(x)$ is not degenerate and the $2$-level drwaing $D(x)$ of $G$ is quasi-planar. 
TLQP has 
$O(m^6)$ circuit complexity, $O(m^4)$ depth, and $O(m^2)$ width. 
\end{lemma}

\paragraph{\bf Proof of \cref{le:gate-TLQP}.} 
Gate TLQP executes three gates: {\sc TL-cross finder} $U_\chi$, {\sc quasi-planarity tester} $U_{Q}$, and {\sc final check} $U_{fc}$, followed by their inverse gates $U_{Q}^{-1}$ and $U_{\chi}^{-1}$. Refer to \cref{fig:TLQP-oracle}.
\begin{figure}[tb!]
    \centering
    \includegraphics[page = 33, width = .7\textwidth]{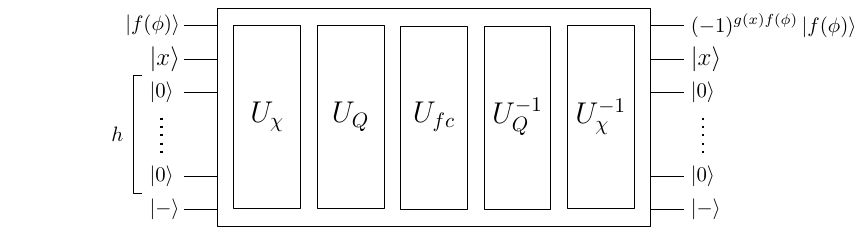}
    \caption{TLQP Oracle Pipeline.}
    \label{fig:TLQP-oracle}
\end{figure}

\paragraph{\sc tl-cross finder.} For the definition of gate $U_\chi$, refer to the proof of \cref{le:gate-TLCM}. Recall that the purpose of $U_{\chi}$ is to compute the crossings in $D(x)$ (under the assumption that $D(x)$ is not degenerate), determined by the vertex order corresponding to $x$. Also recall that, when provided with the input superposition $\ket{x}\ket{0_k}$, the gate $U_{\chi}$ produces the output superposition $\ket{x}\ket{\chi}$.

\paragraph{\sc quasi-planarity tester.}

The purpose of gate $U_Q$ is to verify the absence of any three edges that pairwise cross in $D(x)$; refer to \cref{fig:U_Q}. When provided with the input the superposition $\ket{\chi}\ket{0_h}\ket{0}$, where 
$h \in O(m^4)$, 
the gate $U_Q$ produces the output superposition $\ket{\chi}\ket{0_h}\ket{g(x)}$, where $g(x)=1$ if $D(x)$ is not degenerate and there are not three edges that pairwise cross in $D(x)$.

The gate $U_{Q}$ exploits the auxiliary gate $U_{q\chi}$, whose purpose is to check if three edges pairwise cross; refer to \cref{fig:U_q_chi}. When provided with the input superposition $ \ket{\chi_{i,j}} \ket{ \chi_{i,k}} \ket{\chi_{j,k}}\ket{0}$, the gate provide the output superposition $\ket{\chi_{i,j}} \ket{ \chi_{i,k}} \ket{\chi_{j,k}}\ket{q\chi_{i,j,k}}$, where $q\chi_{i,j,k} = \chi_{i,j} \wedge \chi_{i,k} \wedge \chi_{j,k}$ (which is $1$ if and only if $e_i$, $e_j$, and $e_k$ pairwise cross in $D(x)$). It is implemented using a Toffoli gate with four inputs and outputs, which is activated when $\chi_{i,j}=\chi_{i,k}=\chi_{j,k}=1$. The circuit complexity, depth, and width of $U_{q\chi}$ is $O(1)$.

The gate $U_Q$ works as follows. Consider that if three variables $\chi_{i,j}$, $\chi_{i,k}$, and $\chi_{j,k}$ are compared to determine whether the edges $e_i$, $e_j$, and $e_k$ pairwise cross, none of these variables can be compared with another variable at the same time. Therefore, we partition the pairs of such variables using \cref{le:partion-k} (with $k=3$ and $|X| = \frac{m(m-1)}{2}$) into $p \in O(m^4)$ cross-independent sets $S_1,\dots,S_p$ each containing at most $\frac{m(m-1)}{6}$ unordered triples. 
For $i=1,\dots,p$, the gate $U_Q$ executes in parallel a gate $U_{q\chi}$ for each triple $\{\chi_{i,j}, \chi_{i,k}, \chi_{j,k}\}$ in $S_i$ (refer to \cref{fig:U_Q}). All the last output qubits of the $U_{q\chi}$ gates in $S_i$ enter a Toffoli gate that outputs a qubit $\ket{q_i}$ such that $q_i = 1$ if and only if all of such qubits are equal to $\ket{0}$, i.e., it does not exist a triple of edges that pairwise cross. In order to allow the reuse of the ancilla qubit, except for the qubit $\ket{q_i}$, gate $U_Q$ executes in parallel a gate $U_{q\chi}^{-1}$ for each triple in $S_i$.
All the qubits $\ket{q_i}$, with $i = 1,\dots,p$, enter in cascade a Toffoli gate, with three inputs and outputs, that checks that all of them are equal to $\ket{1}$, i.e., there exist not three edges that pairwise cross. The output qubit of the last Toffoli gate is the qubit $\ket{g(x)}$.
The gate $U_{Q}$ has circuit complexity $O(m^6)$, $O(m^4)$ depth, and $O(m^2)$ width.

\begin{figure}[tb!]
    \centering
        \includegraphics[page = 26, width = .2\textwidth]{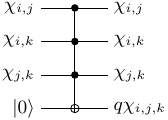}
        \caption{Gate $U_{q\chi}$}
        \label{fig:U_q_chi}
        \end{figure}
        \begin{figure}[tb!]  
            \centering

        \includegraphics[page = 55, width = .95\textwidth]{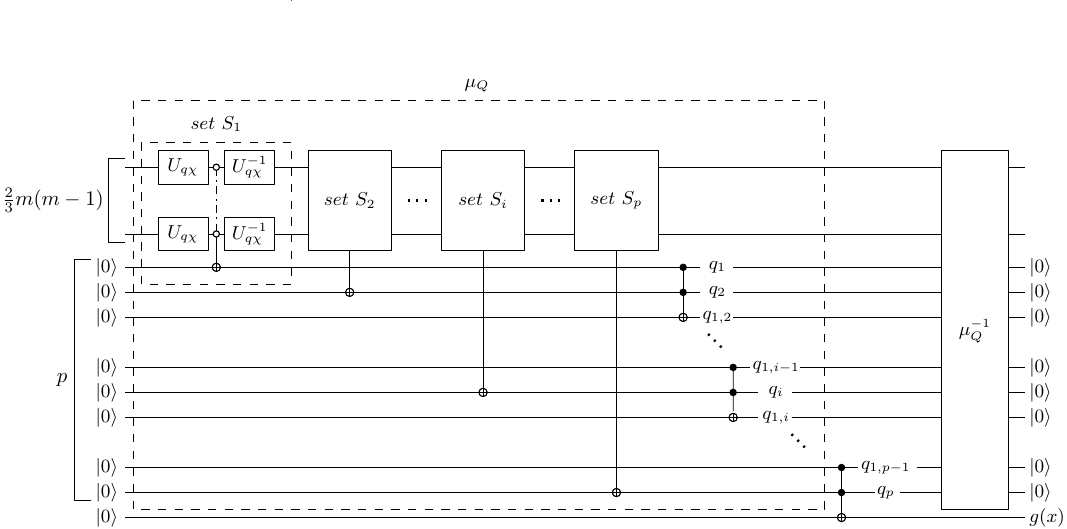}
        \caption{Gate $U_{Q}$}
        \label{fig:U_Q}
\end{figure}

\paragraph{\sc final check.}

We use a gate $U_{fc}$ to check whether the current solution is admissibile, i.e., whether the $2$-level drawing $D(x)$ of $G$ is not degenerate and the $2$-level drawing $D(x)$ of $G$ is quasi-planar. Refer to \cref{fig:final-check} and to \cref{le:gate-TLCM}. To this aim, we provide the 
$U_{fc}$ with the qubit $\ket{f(\phi)}$, a qubit $\ket{-}$, and the qubit $\ket{g(x)}$ provided by gate $U_Q$.
Recall that, gate $U_{fc}$ has $O(1)$ circuit complexity, depth, and width.

\paragraph{The inverse circuits.} The purpose of circuits $U^{-1}_{Q}$ and $U^{-1}_{\chi}$ is to restore the $h$ ancilla qubit to $\ket{0}$ so that they can be used in the subsequent steps of Grover's approach.

\paragraph{Correctness and complexity.}
For the correctness of \cref{le:gate-TLQP}, observe that the gates $U_\chi, U_Q$, and $U_{fc}$ verify all the necessary conditions for which $D(x)$ is a quasi-planar drawing of $G$, under the assumption that $D(x)$ is not degenerate. Therefore, the sign of the output superposition of gate TLQP, which is determined by the expression $(-1)^{g(x)f(\phi)}$, is positive when either $D(x)$ is degenerate or $D(x)$ is not degenerate and it is not a quasi-planar drawing of $G$, and it is negative when $D(x)$ is not degenerate and $D(x)$ is a quasi-planar drawing of $G$. 
The bounds on the circuit complexity, depth, and width descend from the circuit complexity, depth, and width of the gate $U_{Q}$.\qed

\subsection{Problem TLS} 
Recall that the purpose of the basis state $\ket{\theta}$ is to represent a subset $K(\theta)$ of the edges of $G$ of size at most~$\sigma$, each labeled with an integer in $[m]$, and that we denote by $N(\theta)$ the set of indices of the edges in $K(\theta)$.
Also, recall that we denote by $D(x)$ the $2$-level drawing of $G$ associated with the vertex order corresponding to $x$. In the following, for a subgraph $G'$ of $G$, we use the notation $D(x,G')$ to denote the $2$-level drawing of $G'$ induced by $D(x)$.

We call TLS the {\sc Solution Detector} circuit for problem  TLS. 

\begin{lemma}\label{le:gate-TLS}
There exists a gate TLS that, provided with the input superposition
$\ket{f(\phi)}\ket {f(\theta)} \ket{x} \ket{e}\ket{0_{h}} \ket{-}$, where $h \in O(m^2)$, produces the output superposition 
$(-1)^{g(x)f(\phi)f(\theta)}\ket{f(\phi)} \ket {f(\theta)} \ket{x} \ket{e} \ket{0_{h}}\ket{-}$, such that, if $D(x)$ and $N(\theta)$ are not degenerate, then $g(x) = 1$ if and only if $D(x,G')$ is planar, where $G'=(V, E \setminus K(\theta))$.
Gate TLS has  $O(m^2)$ circuit complexity, $O(m)$ depth, and $O(m)$ width.
\end{lemma}

\paragraph{\bf Proof of \cref{le:gate-TLS}.} 
Gate TLS uses three gates: {\sc TL-cross finder} $U_{\chi}$, {\sc skewness cross tester} $U_{sk}$, and {\sc final check} $U_{fc}$, followed by the inverse gates $U^{-1}_{sk}$ and $U^{-1}_{\chi}$.
\cref{fig:TLS-oracle}

\begin{figure}[tb!]
    \centering
    \includegraphics[page = 32, width = .8\textwidth]{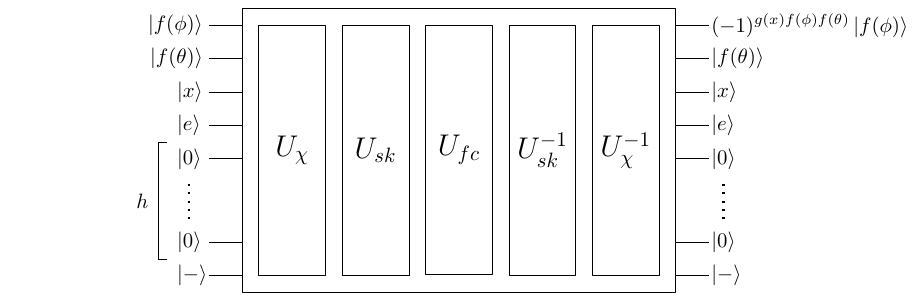}
    \caption{TLS Oracle Pipeline.}
    \label{fig:TLS-oracle}
\end{figure}

\paragraph{\sc TL-cross Finder.}

For the definition of gate $U_\chi$, refer to the proof of \cref{le:gate-TLCM}. Recall that the purpose of $U_{\chi}$ is to compute the crossings in $D(x)$ (under the assumption that $D(x)$ is not degenerate), determined by the vertex order corresponding to $x$. Also recall that, when provided with the input superposition $\ket{x}\ket{0_k}$, the gate $U_{\chi}$ produces the output superposition $\ket{x}\ket{\chi}$. 

\paragraph{\sc skewness cross tester.}
Consider the subgraph $G'$ of $G$ obtained by removing from $G$ all the edges in $K(\theta)$.
The purpose of gate $U_{sk}$ is to determine which of the crossings stored in $\ket{\chi}$ involve pairs of edges that are both absent from $\ket{e}$. In fact, all the edges not in $\ket{e}$ form the edge set of $G'$. Therefore, $U_{sk}$ verifies whether the $2$-layer drawing $D(x,G')$ of $G'$ is planar; refer to \cref{fig:U_SK}. When provided with the input superposition $\ket{\chi}\ket{e}\ket{0_{\frac{m}{2}+p}}\ket{0}$, where $p\in O(m)$, the gate $U_{sk}$ produces the output superposition $\ket{\chi}\ket{e}\ket{0_{\frac{m}{2}+p}}\ket{g(x)}$, 
such that, if $D(x)$ and $N(\theta)$ are not degenerate, then $g(x) = 1$ if and only if $D(x,G')$ is planar.

The gate $U_{sk}$ exploits the auxiliary gate $U_{k}$, whose purpose is to check if any two edges in $G'$ cross; refer to \cref{fig:U_kappa}. When provided with the input superposition $\ket{e_a}\ket{e_b}\ket{\chi_{a,b}}\ket{0}$, the gate $U_k$ provides the output superposition $\ket{e_a}\ket{e_b}\ket{\chi_{a,b}}\ket{s_{a,b}}$, where $s_{a,b}=\neg e_a \wedge \neg e_b \wedge \chi_{a,b}$ (which is $1$ if and only if $e_a$ and $e_b$ cross in $D(x,G')$). The gate $U_k$ is implemented using a Toffoli gate with four inputs and outputs.
The control qubits are $\ket{e_a}$, $\ket{e_b}$, and $\ket{\chi_{a,b}}$.
The target qubit is set to $\ket{0}$. 
The gate is activated when $e_a=e_b=0$ and $\chi_{a,b}=1$. The circuit complexity, depth, and width of $U_{k}$ is $O(1)$.

\begin{figure}[tb!]
    \centering
        \includegraphics[page = 40, width = .2\textwidth]{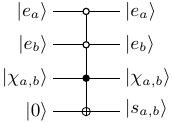}
        \caption{Gate $U_{k}$}
        \label{fig:U_kappa}
\end{figure}

The gate $U_{sk}$ works as follows. Consider that if two variables $e_a$ and $e_b$ are compared to determine whether the corresponding edges cross in $D(x,G')$, none of these variables can be compared with another variable at the same time. Therefore, we partition the pairs of such variables using \cref{le:partion-k} (with $k=2$ and $|X|=m$) into $p \in O(m)$ cross-independent sets $S_1,\dots,S_p$ each containing at most $\frac{m}{2}$ pairs. For $i=1,\dots,p$, the gate $U_{sk}$ executes in parallel a gate $U_k$, for each pair of variables ${e_a,e_b}$ in $S_i$ (together with the corresponding qubit $\chi_{a,b}$), in order to output the qubit $\ket{\neg e_a \wedge \neg e_b \wedge \chi_{a,b}}$. All the last output qubits of the $U_k$ gates in $S_i$ enter a Toffoli gate that outputs a qubit $\ket{sk_i}$ such that $sk_i=1$ if and only if all of them are equal to $\ket{0}$, i.e., there exist no two crossing edges among the pairs in $S_i$. In order to allow the reuse of the ancilla qubit, except for the qubit $\ket{sk_i}$, gate $U_{sk}$ executes in parallel a gate $U^{-1}_{k}$ for each pair in $S_i$. 
To check if all qubits $\ket{sk_i}$ are equal to $\ket{1}$, for $i=1,\dots,p$, we use a series of Toffoli gates $T_i$, each with three inputs and outputs.
The first Toffoli gate $T_1$ receives in input the qubits $\ket{sk_1}$, $\ket{sk_2}$, and a qubit set to $\ket{0}$, and outputs the qubit $\ket{sk_{1,2}} = \ket{sk_{1} \wedge sk_{2}}$. For $i=2,\dots,p$, the Toffoli gate $T_i$ receives in input the qubits $\ket{sk_{1,i-1}}$, $\ket{sk_i}$, and a qubit set to $\ket{0}$, and outputs the qubit $\ket{sk_{1,i}} = \ket{sk_{1,i-1} \wedge sk_{i}}$. 
The output qubit $\ket{sk_{1,p}}$ of the last Toffoli gate $T_p$ is the qubit $\ket{g(x)}$. The gate $U_{sk}$ has circuit complexity $O(m^2)$, depth $O(m)$, and width $O(m)$.

\begin{figure}[tb!]  
\centering
        \includegraphics[page = 54, width = .95\textwidth]{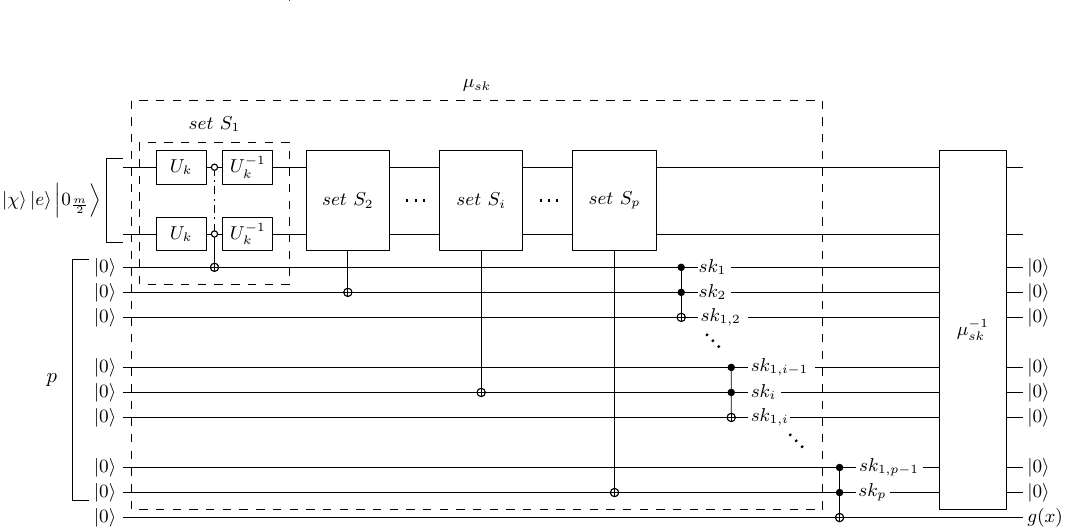}
        \caption{Gate $U_{sk}$.}
        \label{fig:U_SK}
\end{figure}

\paragraph{\sc final check.}

\begin{figure}[tb!]
    \centering
    \includegraphics[page = 52, width = .45\textwidth]{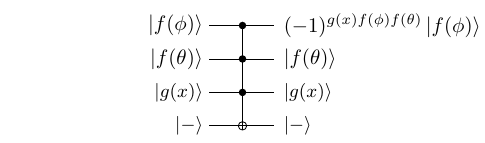}
    \caption{Gate $U_{fc}$.}
    \label{fig:final-check-s}
\end{figure}

The purpose of gate $U_{fc}$ is to check whether the current solution is admissible, i.e., whether the $2$-level drawing $D(x)$ and the set of indices $N(\theta)$ are both not degenerate and the $2$-level drawing $D(x,G')$ of $G'$ is planar. See~\cref{fig:final-check-s}. When provided with the input superposition $\ket{f(\phi)}\ket{f(\theta)}\ket{g(x)}\ket{-}$, the gate $U_{fc}$ produces the outputs superposition $(-1)^{g(x)f(\phi)f(\theta)}\ket{f(\phi)}\ket{f(\theta)}\ket{g(x)}\ket{-}$. Gate $U_{fc}$ exploits a Toffoli gate with four inputs and outputs. The control qubits are $\ket{f(\phi)}$, $\ket{f(\theta)}$, and $\ket{g(x)}$, and the target qubit is $\ket{-}$.
When at least one of $f(\phi)$, $f(\theta)$, and $g(x)$ are equal to $0$, the target qubit leaves unchanged. On the other hand, when $f(\phi)=f(\theta)=g(x)=1$, the target qubit is transformed into the qubit $-\ket{-}$. %
Gate $U_{fc}$ has $O(1)$ circuit complexity, depth, and width. 

\paragraph{The inverse circuits.} The purpose of circuits  $U^{-1}_{sk}$ and $U^{-1}_{\chi}$ is to restore the $h$ ancilla qubit to $\ket{0}$ so that they can be used in the subsequent steps of Grover's approach.

\paragraph{Correctness and complexity.}
For the correctness of \cref{le:gate-TLS}, observe that the gates $U_\chi$, $U_{sk}$, and $U_{fc}$ verify all the necessary conditions for which $D(x,G')$ is a $2$-level planar drawing of $G'$, under the assumption that $D(x)$ and $N(\theta)$ are not degenerate. Therefore, the sign of the output superposition of gate TLS, which is determined by the expression $(-1)^{g(x)f(\phi)f(\theta)}$, is defined as follows. It is
positive when either $D(x)$ or $N(\theta)$ are degenerate or $D(x)$ and $N(\theta)$ are not degenerate and the drawing $D(x,G')$ of $G'$ is not planar. It is negative when $D(x)$ and $N(\theta)$ are not degenerate and the $2$-level drawing $D(x,G')$ of $G'$ is planar. 
The bounds on the circuit complexity, depth, and width of gate TLS descend from those of gate $U_{sk}$.\qed

\subsection{Problem OPCM} We call OPCM the {\sc Solution Detector} circuit for problem  OPCM. Recall that, for the OPCM problem, we denote by $\rho$ the maximum number of crossings allowed in the sought $1$-page layout of $G$. Also, recall that we denote by $\Pi(x)$ the vertex order along the spine of a book layout of $G$ defined by the vertex order corresponding to $x$.

\begin{lemma}\label{le:gate-OPCM}
There exists a gate OPCM that, when provided with the input superposition 
$\ket{f(\phi)}\ket{x}\ket{0_{h}} \ket{-}$, where 
$h \in O(m^2)$, 
produces the output superposition 
$(-1)^{g(x,f(\phi))}\ket{f(\phi)}\ket{x}\ket{0_{h}}\ket{-}$, where $g(x)=1$ if $\Pi(x)$ is not degenerate and the $1$-page layout of $G$ defined by $\Pi(x)$ has at most $\rho$ crossings. OPCM has $O(n^8)$ circuit complexity, $O(n^6)$ depth, and $O(m^2)$ width. 
\end{lemma}

\paragraph{\bf Proof of \cref{le:gate-OPCM}.} 
Gate OPCM executes four gates: {\sc OP-cross finder} $U_{\chi_p}$, {\sc cross counter} $U_{cc}$, {\sc cross comparator} $U_{c<}$, and {\sc final check} $U_{fc}$, followed by the inverse gate $U_{cc}^{-1}$, $U_{c<}^{-1}$, and $U_{\chi_p}^{-1}$. Refer to \cref{fig:OPCM-oracle}.

\begin{figure}[tb!]
    \centering
    \includegraphics[page = 43, width = 0.9\textwidth]{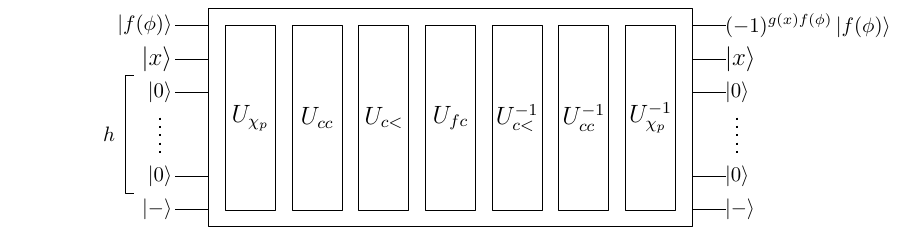}
    \caption{Oracle OPCM.}
    \label{fig:OPCM-oracle}
\end{figure}

\paragraph{\sc op-cross finder.}

The purpose of $U_{\chi_p}$ is to compute the crossings in the $1$-page layout of $G$ defined by $\Pi(x)$, determined by the vertex order corresponding to $x$; refer to \cref{fig:U_ochi_lambda}.  When provided with the input superposition $\ket{x}\ket{0_k}$, where $k=\frac{m(m-1)}{2}$, the gate $U_{p\chi}$ produces the output superposition $\ket{x}\ket{\chi}$.

The gate $U_{\chi_p}$ exploits the auxiliary gate $U_{\lambda}$, whose purpose is to check if two edges cross in the $1$-page layout of $G$ defined by $\Pi(x)$; refer to \cref{fig:U_lambda}. When provided with the input superposition $\ket{x_{i,k}}\ket{x_{k,j}}\ket{x_{j,n}}\ket{x_{i,n}}\ket{0}$, the gate $U_\lambda$ produces the output superposition $\ket{x_{i,k}}\ket{x_{k,j}}\ket{x_{j,n}}\ket{x_{i,n}}\ket{\chi_{a,b}}$, where $e_a = (v_i,v_j)$, $e_b=(v_k,v_n)$, and $\chi_{a,b}$ equals $1$ if and only if $e_a$ and $e_b$ cross in the $1$-page layout of $G$ defined by $\Pi(x)$; refer to \cref{fig:book-crossing} and to the expression $\chi_{a,b}$ in \cref{sse:bool-embeddings}. It is implemented using eight Toffoli gates, each with four inputs and outputs. 
In the following, we assume that $i < k < j < \ell$. 
The first is activated when $x_{i,\ell}=x_{j,k}=1$ and $x_{j,\ell}=0$
(see \cref{fig:book-crossing}, top row, first column). 
The second is activated when $x_{i,k}=x_{j,\ell}=1$ and $x_{j,k}=0$
(see \cref{fig:book-crossing}, top row, second column).
The third is activated when $x_{j,\ell}=x_{i,k}=1$ and $x_{i,\ell}=0$ 
(see \cref{fig:book-crossing}, top row, third column).
The fourth is activated when $x_{j,k}=x_{i,\ell}=1$ and $x_{i,k}=0$
(see \cref{fig:book-crossing}, top row, fourth column).
The fifth is activated when $x_{i,k}=1$ and $x_{i,\ell}=x_{j,k}=0$
(see \cref{fig:book-crossing}, second row, first column). 
The sixth is activated when $x_{j,k}=1$ and $x_{j,\ell}=x_{i,k}=0$
(see \cref{fig:book-crossing}, second row, second column).
The seventh is activated when $x_{i,l}=1$ and $x_{i,k}=x_{j,\ell}=0$
(see \cref{fig:book-crossing}, second row, third column).
The eighth is activated when  $x_{j,\ell}=1$ and $x_{j,k}=x_{i,l}=0$
(see \cref{fig:book-crossing}, second row, fourth column).

The gate $U_{\chi_p}$ works as follows. Consider that if four variables $x_{i,k}, x_{k,j}, x_{j,n}$ and $ x_{i,n}$ are compared to determine whether the edges $(v_i, v_j)$ and $(v_k,v_n)$ cross, none of these variables can be compared with another variable at the same time. Therefore, we partition the pairs of such variables using \cref{le:partion-k} (with $k=4$ and $|X| = \frac{n(n-1)}{2}$) into $r \in O(n^6)$ cross-independent sets $s_1,\dots, S_r$ each containing at most $\frac{n(n-1)}{8}$ pairs. For $i=1,\dots,r$, the gate $U_{\chi_p}$ executes in parallel a $U_\lambda$ gate, for each quartet $(x_{i,k}, x_{k,j}, x_{j,n}$,$ x_{i,n})$ in $S_i$ (refer to \cref{fig:U_pchi}), in order to outputs the qubit $\ket{\chi_{a,b}}$. $U_{\chi_p}$ has circuit complexity $O(n^8)$, depth complexity $O(n^6)$ and width complexity $O(n^2)$.

\begin{figure}[tb!]
        \begin{subfigure}[b]{0.30\textwidth}  
        \includegraphics[page = 16, scale=.8]{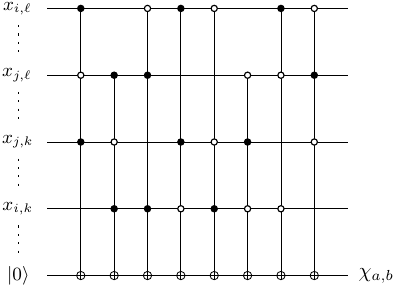}
        \caption{Gate $U_{\lambda}$}
        \label{fig:U_lambda}
        \end{subfigure}
        \hspace{7mm}
        \begin{subfigure}[b]{0.30\textwidth}   
        \includegraphics[page = 47]{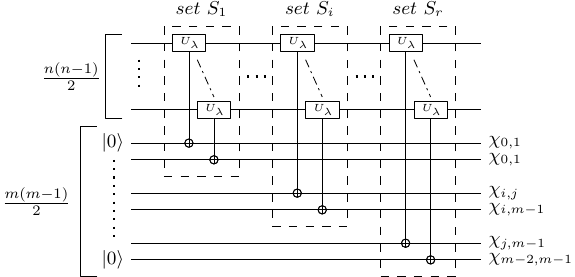}
        \caption{Gate $U_{\chi_p}$}
        \label{fig:U_pchi}
        \end{subfigure}
        \caption{The gate $U_{\lambda}$ (left) and gate $U_{\chi_p}$ (right).}
        \label{fig:U_ochi_lambda}
\end{figure}

\paragraph{\sc cross counter.}

The purpose of gate $U_{cc}$, as mention earlier, is count the total number of crossings in the $1$-page layout of $G$ defined by $\Pi(x)$; refer to \cref{fig:U_s1}. Recall that, when provided with the input superposition $\ket{\chi}\ket{0_h}\ket{0_k}$, where $h=\log m + \log(m-1)$ and $k=2m(m-1) -2(\log m + \log(m-1))-1$, the gate $U_{cc}$ produces the output superposition $\ket{\chi}\ket{\sigma(x)}\ket{0_h}$.

\paragraph{\sc cross comparator.}

The purpose of gate $U_{c<}$, as mention earlier, is to verify if the total number of crossings $\sigma(x)$ in the $1$-page layout of $G$ defined by $\Pi(x)$ compute by the gate $U_{cc}$ is less than the allowed number of crossings $\rho$; refer to \cref{fig:less}. Recall that, when provided with the input superposition $\ket{\sigma(x)}\ket{\rho}\ket{0_h}\ket{0}$, where $h=\log m + \log(m-1)$, the gate $U_{c<}$ produces the output superposition $\ket{\sigma(x)}\ket{\rho}\ket{0_h}\ket{g(x)}$, where $g(x)=1$ if $\Pi(x)$ is not degenerate and $\sigma(x)<\rho$.

\paragraph{\sc final check.}
The purpose of gate $U_{fc}$ is to check whether the current solution is admissible, i.e., whether $\Pi(x)$ is not degenerate and the $1$-page layout of $G$ defined by $\Pi(x)$ has at most $\rho$ crossings. Refer to \cref{fig:final-check}. When provided with the input superposition $\ket{f(\phi)}\ket{g(x)}\ket{-}$, the gate $U_{fc}$ produces the outputs superposition $(-1)^{g(x)f(\phi)}\ket{f(\phi)}\ket{g(x)}\ket{-}$. $U_{fc}$ exploits a Toffoli gate with three inputs and outputs. The control qubits are $\ket{f(\phi)}$ and $\ket{g(x)}$, and the target qubit is $\ket{-}$.
When at least one of $f(\phi)$ and $g(x)$ are equal to $0$, the target qubit leaves unchanged. On the other hand, when $f(\phi)=g(x)=1$, the target qubit is transformed into the qubit $-\ket{-}$. %
Gate $U_{fc}$ has $O(1)$ circuit complexity, depth, and width. 

\paragraph{The inverse circuits.} The purpose of circuits $U^{-1}_{c<}$, $U^{-1}_{cc}$, and $U^{-1}_{\chi}$ is to restore the $h$ ancilla qubit to $\ket{0}$ so that they can be used in the subsequent steps of Grover's approach.

\paragraph{Correctness and complexity.}
For the correctness of \cref{le:gate-OPCM}, observe that the gates $U_\chi$, $U_{cc}$, $U_{c<}$, and $U_{fc}$ verify all the necessary conditions for which the $1$-page layout of $G$ defined by $\Pi(x)$ has at most $\rho$ crossings, under the assumption that $\Pi(x)$ is not degenerate. Therefore, the sign of the output superposition of gate OPCM, which is determined by the expression $(-1)^{g(x)f(\phi)}$, is positive when either $\Pi(x)$ is degenerate or $\Pi(x)$ is not degenerate and the number of crossings $\sigma(x)$ in the $1$-page layout of $G$ defined by $\Pi(x)$ is larger than $\rho$, and it is negative only if $\Pi(x)$ is not degenerate and the number of crossings $\sigma(x)$ in the $1$-page layout of $G$ defined by $\Pi(x)$ is smaller than $\rho$. 
The bound on the circuit complexity descends from the circuit complexity of the gate $U_{\chi_p}$, the bound on the depth descends from the depth of the gate $U_{\chi_p}$, and the bound on the width descends from the width of $U_{cc}$.\qed

\subsection{Problem BT} We call BT the {\sc Solution Detector} circuit for problem  BT. Recall that, for the BT problem, we denote by $\tau$ the number of pages allowed in the sought book layout drawing of $G$.

Recall that, during the computation, we manage the superposition $\ket{\Psi} = \sum_{\psi \in {\mathbb B}^{m \log\tau}} c_\psi \ket{\psi}$ whose purpose is to represent a coloring of the edges of $G$ with colors in the set $[\tau]$. Specifically, consider any basis state $\psi$ that appears in $\ket{\Psi}$. We denote by $P(\psi)$ the page assignment of the edges of $G$ to $\tau$ pages in which, for $i=0,\dots,m-1$, the edge $e_i$ is assigned to the page $\psi[i]$.

\begin{lemma}\label{le:gate-BT}
There exists a gate BT that, when provided with the input superposition 
$\ket{f(\phi)}\ket{\psi}\ket{x}\ket{0_{h}} \ket{-}$, where 
$h \in O(m^2)$,
produces the output superposition 
$(-1)^{g(x,\psi)f(\phi)}\ket{f(\phi)}\ket{\psi}\ket{x}\ket{0_{h}}\ket{-}$, where $g(x,\psi) = 1$ if $\Pi(x)$ is not degenerate and there exists a book layout of $G$ on $\tau$ pages in which the vertex order is $\Pi(x)$ and the page assignment is $P(\psi)$. Gate BT has $O(n^8)$ circuit complexity, $O(n^6)$ depth, and $O(m)$ width. 
\end{lemma}

\paragraph{\bf Proof of \cref{le:gate-BT}.} 
Gate BT uses two gates: {\sc OP-cross finder} $U_{\chi_p}$, {\sc color tester} $U_\beta$, and {\sc final check} $U_{fc}$, followed by the inverse gates $U_{\beta}^{-1}$ and $U_{\chi_p}^{-1}$. Refer to \cref{fig:BT-oracle}.

\begin{figure}[tb!]
    \centering
    \includegraphics[page = 31, width = .7\textwidth]{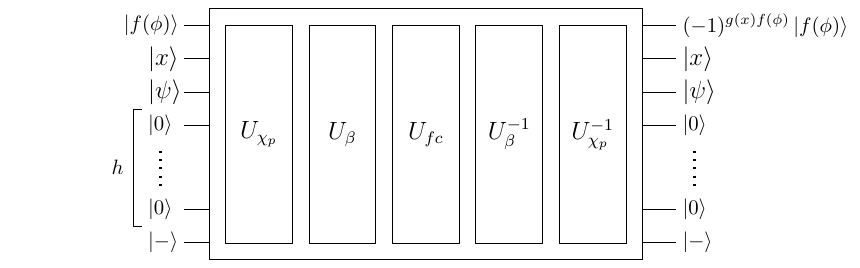}
    \caption{BT Oracle Pipeline.}
    \label{fig:BT-oracle}
\end{figure}

\paragraph{\sc op-cross finder.} 
For the definition of gate $U_{\chi_p}$, refer to the proof of~\cref{le:gate-OPCM}. Recall that the purpose of $U_{\chi_p}$ is to compute the crossings in the $1$-page layout of $G$ defined by $\Pi(x)$, determined by the vertex order corresponding to $x$. Also recall that, when provided with the input superposition $\ket{x}\ket{0_k}$, the gate $U_{p\chi}$ produces the output superposition $\ket{x}\ket{\chi}$.

\paragraph{\sc color tester.}
Consider the book layout $D(x,\psi)$ of $G$ defined by the vertex order $\Pi(x)$ and the page assignment is $P(\psi)$.
The purpose of gate $U_{\beta}$ is to verify if $D(x,\psi)$ is a book embedding of $G$ on $\tau$ pages, provided that $\Pi(x)$ is not degenerate; refer to~\cref{fig:U_beta}. When provided with the input superposition $\ket{\psi}\ket{\chi}\ket{0_b}\ket{0_r}\ket{0}$, where $b = \frac{m}{2}(2+\log \tau)$ and $r=m-1$, the gate $U_\beta$ produces the output superposition $\ket{\psi}\ket{\chi}\ket{0_b}\ket{0_r}\ket{g(x,\psi)}$, where $g(x,\psi)=1$ if $\Pi(x)$ is not degenerate and $D(x,\psi)$ is a book embedding of $G$ on $\tau$ pages.

The gate $U_\beta$ exploits the auxiliary gate $U_{=\lambda}$, whose purpose is to check if two edges that cross have the same color; refer to~\cref{fig:U_=_lambda}. When provided with the input superposition 
$$\ket{\psi[a][0]}\dots\ket{\psi[a][\log(\tau)-1]}\ket{\psi[b][0]}\dots\ket{\psi[a][\log(\tau)-1]}\ket{0_{\log \tau}}\ket{0}\ket{\chi_{a,b}}\ket{0},$$ the gate produces the output superposition $$\ket{\psi[a][0]}\dots\ket{\psi[a][\log(\tau)-1]}\ket{\psi[b][0]}\dots\ket{\psi[a][\log(\tau)-1]}\ket{0_{\log \tau}}\ket{\psi[a]=\psi[b]}\ket{\chi_{a,b}}\ket{\chi_{a,b}\wedge (\psi[a]=\psi[b])}.$$ Gate $U_{=\lambda}$ exploits the auxiliary gates $U_=$ to compare $\psi[a]$ and $\psi[b]$, and a Toffoli gate with two inputs and outputs to verify if edges $e_a$ and $e_b$ cross and have the same color. By \cref{le:gate-equal}, gate $U_{=\lambda}$ has $O(\log \tau)$ circuit complexity, depth, and width.

The gate $U_\beta$ works as follows. Consider that if two variables $\psi[a]$ and $\psi[b]$ are compared to determine whether $e_a$ and $e_b$ have been assigned the same color, none of these variables can be compared with another variable at the same time. Therefore, we partition the pairs of such variables using \cref{le:partion-k} (with $k=2$ and $|X| = m$) into $r \in O(m)$ cross-independent sets $S_1,\dots,S_r$ each containing at most $\frac{m}{2}$ pairs. For $i=1,\dots,r$, the gate $U_{\beta}$ executes in parallel a gate $U_{=\lambda}$, for each pair of variables $\{\psi[a],\psi[b]\} \in S_i$ (together with their corresponding qubit $\chi_{a,b}$), in order to output the qubit $\ket{\chi_{a,b} \wedge (\psi[a]=\psi[b])}$. All the last output qubits of the $U_{=\lambda}$ gates for $S_i$ enter a Toffoli gate that outputs a qubit $\ket{res_i}$ such that $res_i = 1$ if and only if all of them are equal to $\ket{0}$, i.e., it does not exist two crossing edges with the same color (among the pairs in $S_i$). In order to allow the reuse of the ancilla qubits, except for the qubit $\ket{res_i}$, gate $U_\beta$ executes in parallel a gate $U^{-1}_{=\lambda}$ for each pair in $S_i$.
All the qubits $\ket{res_i}$ enter a Toffoli gate that outputs a qubit $\ket{g(x,\psi)}$ such that $g(x,\psi)=1$ if and only if all of them are equal to $\ket{1}$, i.e., there exist no two edges of $G$ with the same color that cross in $D(x,\psi)$.  Gate $U_\beta$ has circuit complexity $O(m^2\log \tau)$, depth $O(m\log \tau)$, and width $O(m)$.

\begin{figure}[tb!]
    \centering
    \includegraphics[page = 24, width = .5\textwidth]{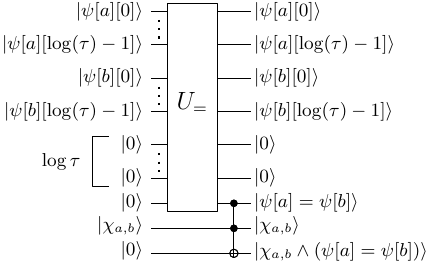}
    \caption{Gate $U_{=\lambda}$.}
    \label{fig:U_=_lambda}
\end{figure}

\begin{figure}[tb!]
    \centering
    \includegraphics[page = 25, width = \textwidth]{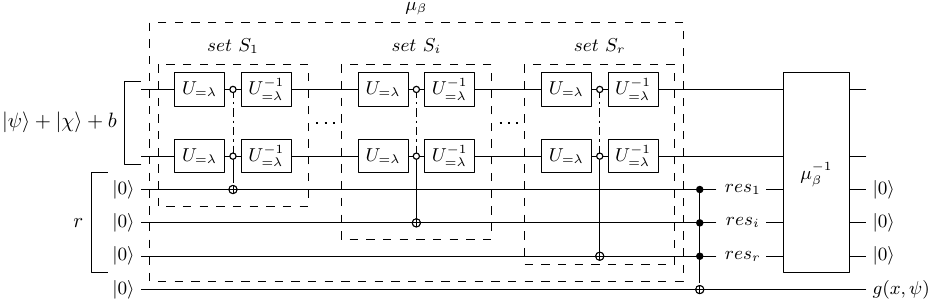}
    \caption{Gate $U_{\beta}$.}
    \label{fig:U_beta}
\end{figure}

\paragraph{\sc final check.}
The purpose of gate $U_{fc}$ is to check whether the current solution is admissible, i.e., whether $D(x,\psi)$ is a book embedding of $G$ on $\tau$ pages. Refer to \cref{fig:U_final-check-bt}. When provided with the input superposition $\ket{f(\phi)}\ket{g(x,\psi)}\ket{-}$, the gate $U_{fc}$ produces the output superposition $(-1)^{g(x,\psi)f(\phi)}\ket{f(\phi)}\ket{g(x,\psi)}\ket{-}$. Gate $U_{fc}$ exploits a Toffoli gate with three inputs and outputs. The control qubits are $\ket{f(\phi)}$ and $\ket{g(x,\psi)}$, and the target qubit is $\ket{-}$.
When at least one of $f(\phi)$ and $g(x,\psi)$ are equal to $0$, the target qubit leaves unchanged. On the other hand, when $f(\phi)=g(x,\psi)=1$, the target qubit is transformed into the qubit $-\ket{-}$. %
Gate $U_{fc}$ has $O(1)$ circuit complexity, depth, and width. 

\begin{figure}[tb!]
    \centering
    \includegraphics[page = 57, width = 0.45\textwidth]{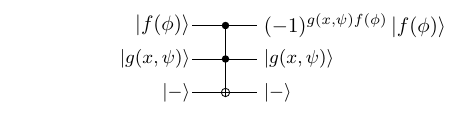}
    \caption{Gate $U_{fc}$.}
    \label{fig:U_final-check-bt}
\end{figure}
\paragraph{The inverse circuits.} The purpose of circuits $U^{-1}_{\beta}$, and $U^{-1}_{\chi_p}$ is to restore the $h$ ancilla qubit to $\ket{0}$ so that they can be used in the subsequent steps of Grover's approach.

\paragraph{Correctness and complexity.}
For the correctness of \cref{le:gate-BT}, observe that the gates $U_{\chi_p}$ and $U_{\beta}$ verify all the necessary conditions for which $D(x,\psi)$ is a book embedding of $G$ with $\tau$ pages, under the assumption that $\Pi(x)$ is not degenerate. Therefore, the sign of the output superposition of gate BT, which is determined by the expression $(-1)^{g(x,\psi)f(\phi)}$, is positive when either $\Pi(x)$ is degenerate or $\Pi(x)$ is not degenerate and the book layout $D(x,\psi)$ is not a book embedding of $G$ with $\tau$ pages, and it is negative only if $\Pi(x)$ is not degenerate and $D(x,\psi)$ is a book embedding of $G$ with $\tau$ pages. 
The bound on the circuit complexity descends from the circuit complexity of gate $U_{\chi_p}$, the bound on the depth descends from the depth of gate $U_{\chi_p}$, and the bound on the width descends from the width of~gate~$U_{\beta}$.\qed

\subsection{Problem BS} 
We call BS the {\sc Solution Detector} circuit for problem  BS.

\begin{lemma}\label{le:gate-BS}
There exists a gate BS that, provided with the input superposition 
$\ket{f(\phi)}\ket {f(\theta)} \ket{x} \ket{e}\ket{0_{h}} \ket{-}$, where $h \in O(m^2)$, produces the output superposition 
$(-1)^{g(x)f(\phi)f(\theta)}\ket{f(\phi)} \ket {f(\theta)} \ket{x} \ket{e} \ket{0_{h}}\ket{-}$, such that, if $\Pi(x)$ and $N(\theta)$ are not degenerate, then $g(x) = 1$ if and only if the $1$-page layout of $G'$ determined by $\Pi(x)$ is a $1$-page book embedding, where $G'=(V, E \setminus K(\theta))$.
Gate BS has 
$O(n^8)$ circuit complexity, $O(n^6)$ depth, and $O(m)$ width.
\end{lemma}

\paragraph{\bf Proof of \cref{le:gate-BS}.} 
Gate BS uses three gates: {\sc OP-cross finder} $U_{\chi_p}$, {\sc Skewness cross tester} $U_{sk}$, and {\sc Final check} $U_{fc}$, followed by the inverse gates $U^{-1}_{sk}$ and $U^{-1}_{\chi_p}$. Refer to \cref{fig:BS-oracle}.

\begin{figure}[tb!]
    \centering
    \includegraphics[page = 42, width = .8\textwidth]{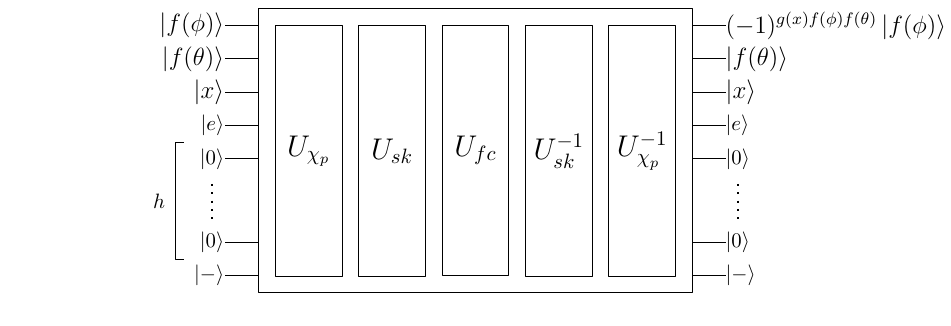}
    \caption{Oracle BS.}
    \label{fig:BS-oracle}
\end{figure}

\paragraph{\sc OP-cross finder.}
For the definition of gate $U_{\chi_p}$, refer to the proof of~\cref{le:gate-OPCM}. Recall that the purpose of $U_{\chi_p}$ is to compute the crossings in the $1$-page layout of $G$ defined by $\Pi(x)$, determined by the vertex order corresponding to $x$. Also recall that, when provided with the input superposition $\ket{x}\ket{0_k}$, the gate $U_{\chi_p}$ produces the output superposition $\ket{x}\ket{\chi}$.

\paragraph{\sc Skewness cross tester.}

For the definition of $U_{sk}$, refer to the proof of~\cref{le:gate-TLS}.
Consider the subgraph $G'$ of $G$ obtained by removing from $G$ all the edges in $K(\theta)$.
The purpose of gate $U_{sk}$ is to determine which of the crossings stored in $\ket{\chi}$ involve pairs of edges that are both absent from $\ket{e}$. In fact, all the edges not in $\ket{e}$ form the edge set of $G'$. Therefore, $U_{sk}$ verifies whether the $1$-page layout of $G'$ determined by $\Pi(x)$ is a $1$-page book embedding. When provided with the input superposition $\ket{\chi}\ket{e}\ket{0_{\frac{m}{2}+p}}\ket{0}$, where $p\in O(m)$, the gate $U_{sk}$ produces the output superposition $\ket{\chi}\ket{e}\ket{0_{\frac{m}{2}+p}}\ket{g(x)}$, 
such that, if $\Pi(x)$ and $N(\theta)$ are not degenerate, then $g(x) = 1$ if and only the $1$-page layout of $G'$ determined by $\Pi(x)$  is a $1$-page book embedding.

\paragraph{\sc Final check.}
For the definition of $U_{fc}$, refer to the proof of~\cref{le:gate-TLS}.
The purpose of gate $U_{fc}$ is to check whether the current solution is admissible, i.e., whether the $1$-page layout of $G$ determined by $\Pi(x)$ and the set of indices $N(\theta)$ are both not degenerate and the $1$-page layout of $G'$ determined by $\Pi(x)$ is a $1$-page book embedding. See \cref{fig:final-check-s}. When provided with the input superposition $\ket{f(\phi)}\ket{f(\theta)}\ket{g(x)}\ket{-}$, the gate $U_{fc}$ produces the outputs superposition $(-1)^{g(x)f(\phi)f(\theta)}\ket{f(\phi)}\ket{f(\theta)}\ket{g(x)}\ket{-}$.

\paragraph{The inverse circuits.} The purpose of circuits  $U^{-1}_{sk}$, and $U^{-1}_{\chi_p}$ is to restore the $h$ ancilla qubit to $\ket{0}$ so that they can be used in the subsequent steps of Grover's approach.

\paragraph{Correctness and complexity.}
For the correctness of \cref{le:gate-BS}, observe that the gates $U_{\chi_p}$, $U_{sk}$, and $U_{fc}$ verify all the necessary conditions for which the $1$-page layout of $G'$ determined by $\Pi(x)$ is a $1$-page book embedding, under the assumption that $\Pi(x)$ and $N(\theta)$ are not degenerate. Therefore, the sign of the output superposition of gate BS, which is determined by the expression $(-1)^{g(x)f(\phi)f(\theta)}$, is defined as follows.
It is positive when either $\Pi(x)$ or $N(\theta)$ are degenerate or $\Pi(x)$ and $N(\theta)$ are not degenerate and the $1$-page layout of $G'$ determined by $\Pi(x)$ is a $1$-page book embedding.
It is negative when $\Pi(x)$ and $N(\theta)$ are not degenerate and the $1$-page layout of $G'$ determined by $\Pi(x)$ is a $1$-page book embedding. 
The bounds on the circuit complexity and depth of gate BS descend from those of $U_{\chi_p}$, whereas the bound on the width of BS descends from $U_{sk}$.\qed

\section{Exploiting Quantum Annealing for Graph Drawing}\label{se:formualtions}

In this section, we explore the $2$-level problems and the book layout problems, that we have addressed so far from the quantum circuit model perspective, in the context of the quantum annealing model of computation. We pragmatically concentrate on the D-Wave platform, which offers quantum annealing services based on large-scale quantum annealing solver. To utilize the hybrid facility of D-Wave for solving an optimization problem, there are essentially two ways: Either the problem is provided with its QUBO formulation or it is provided with a CBO formulation with constraints that are at most quadratic. Also, given a CBO formulation, it is quite simple to construct a QUBO formulation. 
Hence, in~\cref{sse:two-levels,sse:bool-embeddings}, we first provide CBO formulations for the problems introduced in the previous section. Second, we overview (\cref{sse:qcbp-to-qubo}) a standard method for transforming a CBO formulation into a QUBO formulation. Third, in~\cref{sse:dwave}, we discuss a detailed experiment conducted on the quantum annealing services provided by D-Wave, specifically focusing on TLCM, which has extensive experimental literature compared to other problems considered in this paper. These experiments evaluate the efficiency of D-Wave with respect to well-known classical approaches to the TLCM problem.

\subsection{CBO Formulations for Two-Level Problems}\label{sse:two-levels}

Let $G = (U,V,E)$ be a bipartite graph. We denote by $u_i$, for $i=1,\dots, |U|$, and $v_j$, for $j = |U|+1,\dots,|U|+|V|$,  the vertices in $U$ and $V$, respectively. We start by describing the variables and the constraints needed to model the vertex ordering in a $2$-level drawing, which are common to the formulations of TLCM, TLKP, TLQP, and TLS.

\smallskip
\noindent{\bf Ordering variables.} To model the order of the vertices in $U$ and $V$ in a $2$-level drawing $\Gamma$ of $G$, we use $|U|\cdot (|U|-1)$ binary variables $u_{i,j}$ for each ordered pair of vertices $u_i, u_j \in U$ and $|V|\cdot (|V|-1)$ binary variables $v_{i,j}$ for each ordered pair of vertices $v_i, v_j \in V$. The variable $x_{i,j}$ is equal to $1$ if and only if $x_i$ precedes $x_j$ in $\Gamma$, with~$x \in \{u,v\}$. 

\smallskip
\noindent{\bf Ordering constraints.} We define the following constraints.  
As in~\cite{DBLP:journals/jgaa/JungerM97}, to model the fact that an assignment of values to the variables $x_{i,j}$, with $x \in \{u,v\}$, correctly models a linear ordering of the vertices in $U$ and in $V$, we  exploit two types of constraints:
\begin{description}
\item [\sc Consistency:] For each ordered pair of vertices $u_i, u_j \in U$, we have the constraint {\bf (CU)}
$u_{i,j} + u_{j,i} = 1$. Similarly, for each ordered pair of vertices $v_i, v_j \in V$, we have the constraint {\bf (CV)}
$v_{i,j} + v_{j,i} = 1$. Clearly, there exist %
$O(|U|^2)$ and $O(|V|^2)$ 
constraints of type {\bf (CU)} and {\bf (CV)}, respectively. 
\item [\sc Transitivity:] For each ordered triple of vertices $u_i, u_j, u_k \in U$, we have the constraints
{\bf (TU)} $u_{i,j} + u_{j,k} - u_{i,k} \geq 0$ and $u_{i,j} + u_{j,k} - u_{i,k} \leq 1$.
The constraint {\bf (TV)} for each ordered triple of vertices of $V$ is defined analogously. Clearly, there exist $O(|U|^3)$ and $O(|V|^3)$ constraints of type {\bf (TU)} and {\bf (TV)}, respectively.
Constraints {\bf (TU)} and {\bf (TV)} are linear.
We also consider alternative quadratic constraints for transitivity: for each ordered triple of vertices $u_i, u_j, u_k \in U$, we have the constraints
{\bf (TQU)} $1-(u_{i,j} \cdot u_{j,k}) + u_{i,k}\geq 1$.
The constraints {\bf (TQV)} for each ordered triple of vertices of $V$ are defined analogously. Clearly, the number of {\bf (TQU)} and {\bf (TQV)} constraints is half the number of {\bf (TU)} and {\bf (TV)} constraints.
\end{description}

Next, we provide specific variables and constraints that allow us to correctly model the problems TLCM, TLKP, TLQP, and TLS.
To this aim, for each pair of independent edges $e_a=(u_i,v_k)$ and $e_b=(u_j,v_\ell)$, we define the expression $\chi_{a,b} = u_{i,j}\cdot v_{\ell,k} + u_{j,i}\cdot v_{k,\ell}$, which is equal to $1$ if and only if $e_a$ and $e_b$ cross. 
For each edge $e \in E$, we  denote by $I(e)$ the set of edges in $E$ that do not share an endpoint with $e$.

\paragraph{Two-level Crossing Minimization (TLCM).} We consider the minimization version of the problem. In order to minimize the total number of crossings in the sought $2$-level drawing of $G$, we define the objective function {(\bf OF)}
\newcommand{\BBBB}{\min {\sum_{e_a \in E}} ~\sum_{e_b \in I(e_a)} \chi_{a,b}}

\[\BBBB.\]

\paragraph{Two-level $k$-planarity (TLKP).} 
We show how to model the fact that at most $k$ crossings are allowed on each edge. We have the single constraint {\bf (KP)}
\newcommand{\KPPP}{\sum_{e_b\in I(e_a)} \chi_{a,b} \leq k}

\mbox{$\KPPP$.}

\noindent Clearly, over all the edges of $G$, there are $|E|$ constraints of type {\bf (KP)}.

\paragraph{Two-level Quasi Planarity (TLQP).}

We show how to model the fact that no three edges are allowed to pairwise cross. For each ordered triple $(e_a,e_b,e_c)$ of edges of $E$ such that 
$e_b  \in I(e_a)$ and $e_c \in I(e_a) \cap I(e_b)$, 
we have the constraint~{\bf (QP)} 
\newcommand{\CQP}{\chi_{a,b} + \chi_{b,c} + \chi_{a,c} < 3}

\[\CQP.\]

\noindent Clearly, over all the edges of $G$, there are $O(|E|^3)$ constraints of type~{\bf (QP)}.

\paragraph{Two-level Skewness (TLS).} To model the membership of the edges to a subset $S$ such that $|S|\leq \sigma$, whose removal from $G$ yields a forest of caterpillars, we use $|E|$ binary variables $s_{i,j}$ for each edge $(u_i,v_j)$. The variable $s_{i,j}$ is equal to $1$ if and only if $(u_i,v_j)$ belongs to $S$. First, to enforce that $|S| \leq \sigma$, we use the constraint {\bf (CS)} 
\newcommand{\CCS}{\sum_{(u_i,v_j) \in E} s_{i,j} \leq \sigma}

\[\CCS.\]

Second, we show how to model the fact that no two edges in $E \setminus S$ are allowed to cross. For each edge $e_a = (u_i,v_j)\in E$ and for each edge $e_b = (u_\ell,v_k) \in I(e_a)$, we have the constraint {\bf (S)}
\newcommand{\CSS}{\chi_{a,b} - s_{i,j} - s_{\ell,k} < 1}

\[\CSS.\]

Over all the edges of $G$ there are $O(|E|^2)$ constraints of type {\bf (S)}.

\subsection{CBO Formulations for Book-layout problems}\label{sse:bool-embeddings}

Let $G=(V,E)$ be a graph. We denote by $v_i$, for $i=1,\dots,|V|$, the vertices in~$V$. As in \cref{sse:two-levels}, we use the variable $x_{i,j}$ to encode the precedence between the ordered pair of vertices $v_i, v_j \in V$. Moreover, in order for an assignment of values in ${\mathbb B}$ to such variables to correctly model a linear ordering of $V$, we use the {\em consistency} and {\em transitivity} constraints described in \cref{sse:two-levels}.

Next, we provide the specific variables and constraints that allow us to correctly model the problems OPCM, BT, and BS. Two edges are \emph{independent} if they do not share an end-vertex.
To this aim, for each ordered pair of independent edges $e_a=(v_i,v_j)$ and $e_b=(v_\ell,v_k)$, we define the expression $\chi_{a,b} = x_{i,\ell}\cdot x_{\ell, j}\cdot x_{j,k} + x_{i,k}\cdot x_{k, j}\cdot x_{j,\ell} + x_{j,\ell}\cdot x_{\ell, i}\cdot x_{i,k} + x_{j,k}\cdot x_{k, i}\cdot x_{i,\ell}$, which is equal to $1$ if and only if $e_a$ and $e_b$ cross and (exactly) one of the endpoints of $e_a$ precedes both the endpoints of $e_b$; refer to~\cref{fig:book-crossing}\textcolor{blue}{(top)}. %
More specifically, let $x_{\alpha,\beta} \cdot x_{\beta,\gamma} \cdot x_{\gamma,\delta}$ be any of the four terms that define $\chi_{a,b}$. We have that such a term evaluates to $1$ if and only if the vertices $v_{\alpha}$, $v_{\beta}$, $v_{\gamma}$, and $v_{\delta}$ appear in this left-to-right order along the spine.

\begin{figure}[tb!]
    \centering
    \includegraphics[page = 35, width = \textwidth]{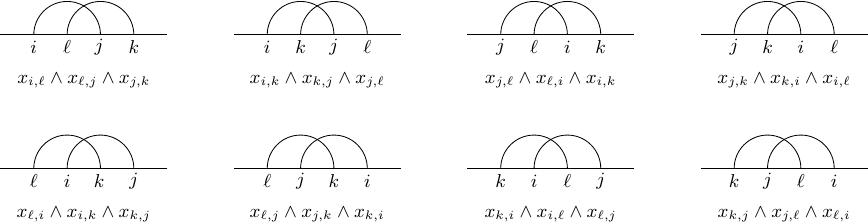}
    \caption{(top) The four possible crossing configurations of the edges $e_a=(v_i,v_j)$ and $e_b=(v_\ell,v_k)$, in which an end-vertex of $e_a$ precedes both endpoints of $e_b$. (bottom) The four possible crossing configurations of the edges $e_a$ and $e_b$, in which an end-vertex of $e_b$ precedes both endpoints of $e_a$.}
    \label{fig:book-crossing}
\end{figure}

\paragraph{One-Page Crossing Minimization (OPCM).}
We consider the minimization version of the problem. In order to minimize the total number of crossings in the sought $1$-page layout of $G$, we define the objective function {(\bf OF)}
\newcommand{\CBB}{\min{\sum_{e_a \in E}} ~\sum_{e_b \in I(e_a)} \chi_{a,b}}

\[\CBB.\]

\paragraph{Book Thickness (BT).}
To model the membership of the edges to one of the $\tau$ pages, we use $\tau|E|$ binary variables $e_{i,j,c}$, for each edge $(v_i,v_j)$ and for each page $c \in [\tau]$. The variable $e_{i,j,c}$ is equal to $1$ if and only if $(v_i,v_j)$ is assigned to page $c$. First, in order to enforce that each edge belongs exactly to one page, for each edge $(v_i,v_j) \in E$, we use the constraint {\bf (BC)} 
\newcommand{\CBC}{\sum_{c \in [\tau]} e_{i,j,c} = 1}

\[\CBC.\]

Second, we show how to model the fact that no two edges assigned to the same page are allowed to cross. For each edge $e_a=(v_i,v_j) \in E$, for each edge $e_b=(v_l,v_k)\in I(e_a)$, and for each page $c \in [\tau]$, we have the constraint~{\bf (CC)}
\newcommand{\CCC}{\chi_{a,b} + e_{i,j,c} + e_{l,k,c} < 3}

\[\CCC.\]

\paragraph{Book Skewness (BS).} For this problem, we adopt the same constraints {\bf (CS)} as for the TLS problem. Moreover, we adopt the constraints {\bf (S)} as for the TLS problem.

\subsection{From CBO to QUBO}\label{sse:qcbp-to-qubo}

The formulations presented in \cref{sse:two-levels,sse:bool-embeddings} contain quadratic and even cubic constraints. Such constraints, however, only involve binary variables. Hence, they can be easily linearized, by means of standard operations research techniques, to be exploited to define a QUBO formulation suitable for quantum annealing. 

Specifically, let $\mu =  \Pi^k_{i=1}x_i$ be a (non-necessarily quadratic) monomial of total degree $k$, such that each variable $x_i$ is a binary variable. We obtain an equivalent constraint by replacing each occurrence of $\mu$ in all the constraints of our formulation with a new binary variable $z_\mu$ and by adding the following $k+1$ constraints:

\[  \left\{ \begin{array}{ll}
         z_\mu \leq x_i & i = 1,\dots,k\\
        z_\mu \geq 1-k + \sum^k_{i=1} x_i & \end{array} \right. \] 

In our formulation for the $2$-level problems, the maximum degree of all monomials is $2$ and these monomials arise from the expressions $\chi_{a,b}$, for each pair of independent edges $e_a = (u_i,v_i)$ and $e_b = (u_\ell,v_k)$. Thus, there exist at most $|E(G)|^2$ distinct degree-$2$ monomials in these formulations. Therefore, by applying the replacement described above, we introduce at most $|E(G)|^2$ new variables and $3|E(G)|^2$ new constraints.
Similarly, in our formulation for the book layout problems, the maximum degree of all monomials is $3$ and these monomials arise from the expressions $\chi_{a,b}$, for each pair of independent edges $e_a = (u_i,v_i)$ and $e_b = (u_\ell,v_k)$. Thus, there exist at most $|E(G)|^2$ distinct degree-$3$ monomials in these formulations. Therefore, by applying the replacement described above, we introduce at most $|E(G)|^2$ new variables and $4|E(G)|^2$ new constraints.

Once the constraints are linearized, a QUBO formulation is obtained by inserting all constraints in the objective function. Note that if the problem is not an optimization problem an objective function is anyway created such that its optimum value is equal to zero.
Constraints are first transformed so that their right member is equal to zero. For inequalities an extra variable is inserted so to transform the inequality into an equality.
Then left member is squared and the result is inserted into the objective function.

\subsection{D-Wave Experimentation}\label{sse:dwave}
We performed our experiments on TLCM using the hybrid solver of D-Wave, which suitably mixes quantum computations with classic tabu-search and simulated annealing heuristics. The obtained results are not guaranteed to be optimal.

The D-Wave hybrid solver receives in input either a CBO or a QUBO formulation of a problem. We used it with a linear CBO formulation (with {\bf (TU)} and {\bf (TV)} constraints) and with a quadratic CBO formulation (with {\bf (TQU)} and {\bf (TQV)} constraints).
Roughly, D-Wave hybrid solver works as follows. First, it decomposes the problem into parts that fit the quantum processor. The decomposition aims at selecting subsets of variables, and hence sub-problems, maximally contributing to the problem energy. Second, it solves such sub-problems with the quantum processor. Third, it injects the obtained results into the original overall problem that is solved with traditional heuristics. These steps can be repeated several times, since an intermediate solution can re-determine the set of variables that contribute the most to the energy of the problem. An interesting description of the behaviour and of the limitations of D-Wave is presented in \cite{DBLP:journals/corr/abs-1904-11965}, although it refers to the quantum processor called Chimera that has been replaced by the new processors called Pegasus and Zephyr.

We compare our results with the figures proposed in~\cite{DBLP:journals/informs/BuchheimWZ10}. Namely, the authors compare three exact algorithms for TLCM: LIN, which is the standard linearization approach; JM, which is the algorithm in~\cite{DBLP:journals/jgaa/JungerM97}; and SDP, which is the branch-and-bound in~\cite{DBLP:conf/ipco/RendlRW07}. Their experiments were carried out on an Intel Xeon processor with 2.33 GHZ.

\cref{tab:d-wave-experiment} illustrates the results of the experimentation we conducted on the D-Wave platform. We focused on the same set of graphs used in \cite{DBLP:journals/informs/BuchheimWZ10}. 
Namely, for each value of $n$, i.e., number of vertices per layer, and for each value of $d$, i.e., density, we performed $10$ experiments on $10$ distinct graph instances. \cref{tab:d-wave-experiment} reports for each pair $n,d$ the averages over such instances. The columns report what follows. {\it Best Current Time}: the best performance among the exact algorithms for TLCM known so far (time measured by \cite{DBLP:journals/informs/BuchheimWZ10}). {\it Approach}: one of LIN, JM, and SDP depending on which is the fastest approach. The double slash indicates that an optimal solution has not been found. {\it D-wave Time (Linear)}: time of our experiments, using the linear formulation. {\it \# of Constrains (Linear)}: number of constraints of our linear formulation. {\it Crossings (Linear)}: number of crossings obtained with our linear formulation. {\it D-wave Time (Quadratic)}: time of our experiments, using the quadratic formulation. {\it \# of Constrains (Quadratic)}: number of constraints of our quadratic formulation. {\it Crossings (Quadratic)}: number of crossings obtained with our quadratic formulation.

It is important to observe that the number of crossings we obtained was the optimal one for all graphs with $10$ vertices per layer \cite{Zheng2007ANE} (graphs above the horizontal line of \cref{tab:d-wave-experiment}). Unfortunately, the literature~\cite{DBLP:journals/informs/BuchheimWZ10} does not report the actual minimum number of crossings for the remaining instances.

The comparison of the time employed by D-Wave (linear and quadratic) with the one of the best exact methods is quite promising, even if the times in the {\it Best Current Time} column are the results of a computation performed on a non-up-to-date classical hardware, and indicate that D-Wave can be used to efficiently tackle instances of TLCM. The comparison between linear and quadratic CBO formulations indicates that the quadratic formulation is more efficient, since it generates fewer constraints. Their behaviour in terms of number of crossings are quite similar. The time we report is the overall time elapsed between the remote call from our client and the reception of the result. The actual time spent on the quantum processor is always between 0.016 and 0.032 sec. 

\begin{table}[!ht] 
    \centering
    \caption{TLCM: D-Wave vs other approaches. Bipartite graphs with increasing number of vertices per layer and density. All times are in seconds.}
    \label{tab:d-wave-experiment}
    \resizebox{\linewidth}{!}{
    \begin{tabular}{|c|c|r|c|r|r|r|r|r|r|r|}
        \hline
        \begin{minipage}[t][1.1cm][c]{0.6cm} \centering $n$ \end{minipage} 
        & 
        \begin{minipage}[t][0.9cm][c]{0.6cm} \centering $d$ \end{minipage} 
        & 
        \begin{minipage}[t][0.9cm][c]{1.3cm} \centering Best Current Time \end{minipage}  
        & 
        \begin{minipage}[t][0.9cm][c]{1.6cm} \centering Approach \end{minipage}
        & 
        \begin{minipage}[t][0.9cm][c]{1.6cm} \centering D-Wave Time (Linear) \end{minipage}
        & 
        \begin{minipage}[t][0.9cm][c]{1.6cm} \centering \# of Constraints (Linear) \end{minipage}
        & 
        \begin{minipage}[t][0.9cm][c]{1.6cm} \centering Crossings (Linear) \end{minipage} 
        & 
        \begin{minipage}[t][0.9cm][c]{1.6cm} \centering D-Wave Time (Quadratic) \end{minipage}
        & 
        \begin{minipage}[t][0.9cm][c]{1.6cm} \centering \# of Constraints (Quadratic) \end{minipage}
        & 
        \begin{minipage}[t][0.9cm][c]{1.6cm} \centering Crossings (Quadratic) \end{minipage}
        \\
        \hline
10 & 10 & 0,01 & LIN & 0,26 & 509 & 1 & 0,15 & 269 & 1 %
\\ 
        10 & 20 & 0,05 & JM & 0,92 & 1926 & 11 & 0,50 & 996 & 11 %
        \\ 
        10 & 30 & 0,15 & JM & 1,48 & 2969 & 52 & 0,79 & 1529 & 52 %
        \\ 
        10 & 40 & 0,33 & JM & 1,39 & 2970 & 142 & 0,76 & 1530 & 142 %
        \\ 
        10 & 50 & 0,61 & JM & 1,44 & 2970 & 276 & 0,83 & 1530 & 276 %
        \\ 
        10 & 60 & 1,14 & JM & 1,52 & 2970 & 259 & 0,81 & 1530 & 459 %
        \\ 
        10 & 70 & 2,35 & JM & 1,47 & 2970 & 717 & 0,79 & 1530 & 717 %
        \\ 
        10 & 80 & 4,05 & JM & 1,49 & 2970 & 1037 & 0,81 & 1530 & 1037 %
        \\ 
        10 & 90 & 6,79 & SDP & 1,92 & 2970 & 1387 & 0,81 & 1530 & 1387 %
        \\ \hline
        12 & 10 & 0,02 & LIN & 1,18 & 2293 & 0,67 & 0,61 & 1183 & 0,67 %
        \\ 
        12 & 20 & 1,52 & JM & 2,15 & 4110 & 34,22 & 1,07 & 2110 & 34,22%
        \\ 
        12 & 30 & 4,53 & JM & 2,77 & 5263 & 139,11 & 1,41 & 2696 & 139,11 %
        \\ 
        12 & 40 & 16,36 & JM & 2,72 & 5337 & 339,89 & 1,44 & 2734 & 339,89%
        \\ 
        12 & 50 & 44,84 & SDP & 2,80 & 5412 & 664,56 & 1,52 & 2772 & 664,56 %
        \\ 
        12 & 60 & 48,26 & SDP & 2,94 & 5412 & 1040 & 1,44 & 2772 & 1040 %
        \\ 
        12 & 70 & 40,31 & SDP & 2,71 & 5412 & 1535 & 1,48 & 2772 & 1535%
        \\ 
        12 & 80 & 28,71 & SDP & 2,82 & 5412 & 2228,67 & 1,63 & 2772 & 2228,56 %
        \\ 
        12 & 90 & 22,21 & SDP & 2,88 & 5412 & 3023,67 & 1,68 & 2772 & 3023,67 %
        \\ 
        14 & 10 & 0,33 & LIN & 2,30 & 3912 & 2,67 & 1,10 & 2008 & 2,67 %
        \\ 
        14 & 20 & 89,61 & SDP & 4,19 & 7701 & 89,33 & 2,37 & 3933 & 89,33 %
        \\ 
        14 & 30 & 132,72 & SDP & 4,69 & 8512 & 316,89 & 2,55 & 4344 & 316,78 %
        \\ 
        14 & 40 & 144,03 & SDP & 5,20 & 8918 & 716 & 2,74 & 4550 & 716 %
        \\ 
        14 & 50 & 180,49 & SDP & 4,71 & 8918 & 1316,11 & 2,88 & 4550 & 1316 %
        \\ 
        14 & 60 & 141,93 & SDP & 4,67 & 8918 & 2053 & 2,35 & 4550 & 2052,89 %
        \\ 
        14 & 70 & 149,68 & SDP & 4,93 & 8918 & 3017,78 & 2,76 & 4550 & 3016,22 %
        \\ 
        14 & 80 & 145,97 & SDP & 5,03 & 8918 & 4258,89 & 2,46 & 4550 & 4257,67 %
        \\ 
        14 & 90 & 81,27 & SDP & 4,92 & 8918 & 5865,33 & 2,39 & 4550 & 5875,22 %
        \\ 
        16 & 10 & 2,77 & LIN & 3,45 & 6672 & 11,56 & 1,83 & 3410 & 11,56 %
        \\ 
        16 & 20 & 309,31 & SDP & 6,70 & 12423 & 176,78 & 3,34 & 6324 & 176,33 %
        \\ 
        16 & 30 & 630,31 & SDP & 7,53 & 13397 & 604,78 & 3,63 & 6817 & 603,89 %
        \\ 
        16 & 40 & 800,87 & SDP & 7,56 & 13680 & 1326,89 & 3,73 & 6960 & 1324,44 %
        \\ 
        16 & 50 & 451,09 & SDP & 7,46 & 13680 & 2376,89 & 3,82 & 6960 & 2375,44 %
        \\ 
        16 & 60 & 403,82 & SDP & 7,33 & 13680 & 3770,11 & 3,73 & 6960 & 3762,67 %
        \\ 
        16 & 70 & 789,62 & SDP & 7,48 & 13680 & 5501,67 & 3,77 & 6960 & 5486,78 %
        \\ 
        16 & 80 & 568,55 & SDP & 7,37 & 13680 & 7591,78 & 3,66 & 6960 & 7575,11 %
        \\ 
        16 & 90 & 362,29 & SDP & 7,08 & 13680 & 10336,11 & 3,72 & 6960 & 10310,89 %
        \\ 
        18 & 10 & 7,06 & LIN & 6,16 & 12154 & 18,11 & 3,36 & 6187 & 18,11 %
        \\ 
        18 & 20 & 778,86 & SDP & 10,25 & 18424 & 312,78 & 4,99 & 9358 & 309,56 %
        \\ 
        20 & 10 & 117,72 & LIN & 10,25 & 19357 & 44,78 & 5,26 & 9828 & 44,33 %
        \\ 
        20 & 20 & 1813,87 & SDP & 14,25 & 26589 & 551,33 & 7,34 & 13479 & 548,22 %
        \\ 
        22 & 10 & 546,71 & LIN & 14,06 & 25942 & 94,22 & 7,09 & 13152 & 92,22 %
        \\ 
        22 & 20 & 3443,81 & SDP & 19,62 & 35217 & 984,22 & 10,04 & 17830 & 962,44 %
        \\ 
        24 & 10 & 2225,82 & LIN & 20,96 & 37757 & 150,44 & 10,69 & 19110 & 148,33 %
        \\ 
        24 & 20 & // & // & 27,70 & 48788 & 1794,22 & 14,17 & 24669 & 1835,11 %
        \\ \hline

    \end{tabular}}
\end{table}

\section{Conclusions and Open Problems}

We initiate the study of quantum algorithms in the Graph Drawing research area, providing a framework that allows us to tackle several classic problems within the $2$-level and book layout drawing standards. Our framework, equipped with several quantum circuits of potential interest to the community, builds upon Grover's quantum search approach. It empowers us to achieve, at least, a quadratic speedup compared to the best classical exact algorithms for all the  problems under consideration.
In addition, we conducted experiments using the D-Wave quantum annealing platform for the {\sc Two-Level Crossing Minimization} problem. Our experiments demonstrated that the platform is highly suitable for addressing graph drawing problems and showcased significant efficiency when compared to the top approaches available for solving such problems. 
The encounter between Graph Drawing and Quantum Computing is still in its nascent stage, offering a vast array of new and promising problems. Virtually, all graph drawing problems can be explored through the lenses of quantum computation, utilizing both the quantum circuit model and, more pragmatically, quantum annealing platforms.

\clearpage
\bibliographystyle{splncs04}
\bibliography{bibliography}
\end{document}